\documentclass[11pt]{article}

\usepackage{amsthm}
\usepackage{algorithm, algorithmic}
\usepackage{graphicx}
\usepackage{amssymb}

\def\denseformat{
\setlength{\textheight}{9in}
\setlength{\textwidth}{6.9in}
\setlength{\evensidemargin}{-0.2in}
\setlength{\oddsidemargin}{-0.2in}
\setlength{\headsep}{10pt}
\setlength{\topmargin}{-0.3in}
\setlength{\columnsep}{0.375in}
\setlength{\itemsep}{0pt}
}

\denseformat
\begin{document}

\newtheorem{thm}{Theorem}[section]
\theoremstyle{definition}
\newtheorem{dfn}{Definition}[section]
\theoremstyle{remark}
\theoremstyle{plain}
\newtheorem{lem}[thm]{Lemma}
\newtheorem{col}[thm]{Corollary}
\newtheorem{fact}[thm]{Fact}
\newtheorem{fig}[figure]{Fig.}

\def\MathN{\hbox{\rm I\kern-2pt I\kern-3.1pt N}}
\def\Expect{\hbox{\rm I\kern-2pt I\kern-3.1pt E}}
\title{A Fast Network-Decomposition Algorithm \\ and its Applications to Constant-Time \\ Distributed Computation}
\author{Leonid Barenboim\thanks{Open University of Israel.
 E-mail: {\tt leonidb@openu.ac.il}. 
 Part of this work has been performed while the author was a postdoctoral fellow at a joint program of the Simons Institute at UC Berkeley and I-Core at Weizmann Institute.} 
 \and Michael Elkin\thanks{Ben-Gurion University of the Negev. Email: {\tt elkinm@cs.bgu.ac.il} 
This research has been supported by the Israeli Academy of Science, grant 593/11, and by the Binational Science Foundation, grant 2008390.} \and Cyril Gavoille\thanks{LaBRI - Universite de Bordeaux. Email: {\tt gavoille@labri.fr}}}

\maketitle
\begin{abstract}

A partition $(C_1,C_2,...,C_q)$ of $G = (V,E)$ into clusters of strong (respectively, weak) diameter $d$, such that the supergraph obtained by contracting each $C_i$ is $\ell$-colorable is called a strong (resp., weak) $(d, \ell)$-network-decomposition. Network-decompositions were introduced in a seminal paper by Awerbuch, Goldberg, Luby and Plotkin in 1989. Awerbuch et al. showed that strong $(exp\{O(\sqrt{ \log n \log \log n})\}$, $exp\{O(\sqrt{ \log n \log \log n})\})$-network-decompositions can be computed in distributed deterministic time $exp\{O(\sqrt{ \log n \log \log n})\}$. Even more importantly, they demonstrated that network-decompositions can be used for a great variety of applications in the message-passing model of distributed computing.

The result of Awerbuch et al. was improved by Panconesi and Srinivasan in 1992: in the latter result $d = \ell = exp\{O(\sqrt{\log n})\}$, and the running time is $exp\{O(\sqrt{\log n})\}$ as well. In another remarkable breakthrough Linial and Saks (in 1992) showed that {weak} $(O(\log n), O(\log n))$-network-decompositions can be computed in distributed randomized time $O(\log^2 n)$. Much more recently Barenboim (2012) devised a distributed randomized constant-time algorithm for computing strong network decompositions with $d = O(1)$. However, the parameter $\ell$ in his result is $O(n^{1/2 + \epsilon})$.

In this paper we drastically improve the result of Barenboim and devise a distributed randomized constant-time algorithm for computing strong $(O(1), O(n^{\epsilon}))$-network-decompositions. As a corollary we derive a constant-time randomized $O(n^{\epsilon})$-approximation algorithm for the distributed minimum coloring problem, improving the previously best-known $O(n^{1/2 + \epsilon})$ approximation guarantee. We also derive other improved distributed algorithms for a variety of problems.

Most notably, for the extremely well-studied distributed minimum dominating set problem currently there is no known deterministic polylogarithmic-time algorithm. We devise a {\em deterministic} polylogarithmic-time approximation algorithm for this problem, addressing an open problem of Lenzen and Wattenhofer (2010). 

\end{abstract}

\pagenumbering {arabic}

\section{Introduction}
{\bf 1.1 Network-Decompositions\\}
In the distributed message-passing model a communication network is represented by an $n$-vertex graph $G = (V,E)$. The vertices of the graph host processors that communicate over the edges. Each vertex has a unique identity number (ID) from the range $\{1,2,...,n\}$. We consider a synchronous setting in which computation proceeds in rounds, and each message sent over an edge arrives by the beginning of the next round. The running time of an algorithm is the number of rounds from the beginning until all vertices terminate. Local computation is free.

A {\em strong} (respectively, {\em weak}) {\em diameter} of a cluster $C \subseteq V$ is the maximum distance $\mbox{dist}_{G(C)} (u,v)$ (resp., $\mbox{dist}_G(u,v)$) between a pair of vertices $u,v \in C$, measured in the induced subgraph $G(C)$ of $C$ (resp., in $G$). A partition $(C_1,C_2,...,C_q)$ of $G = (V,E)$ into clusters of strong (resp., weak) diameter $d$, such that the supergraph ${\cal G} = ({\cal V}, {\cal E})$, ${\cal V} = \{C_1,C_2,...,C_q \}$, ${\cal E} = \{ (C_i,C_j) \ | \ C_i, C_j \in {\cal V}, i \neq j, \exists v_i \in C_i, v_j \in C_j, (v_i,v_j) \in E \}$ obtained by contracting each $C_i$ is $\ell$-colorable is called a {\em strong} (resp., {\em weak}) {\em $(d,\ell)$-network-decomposition}.

Network-decompositions were introduced in a seminal paper by Awerbuch et al. \cite{AGLP89}. The authors of this paper showed that strong $(exp\{O(\sqrt{ \log n \log \log n})\},$ $ exp\{O(\sqrt{ \log n \log \log n})\})$-network-decompositions can be computed in distributed deterministic $exp\{O(\sqrt{ \log \log \log n})\}$ time. Even more importantly they demonstrated that many pivotal problems in the distributed message passing model can be efficiently solved if one can efficiently compute $(d, \ell)$-network-decompositions with sufficiently small parameters. In particular, this is the case for Maximal Independent Set, Maximal Matching, and $(\Delta + 1)$-Vertex-Coloring.

The result of \cite{AGLP89} was improved a few years later by Panconesi and Srinivasan \cite{PS95} who devised a deterministic algorithm for computing strong $(exp\{O(\sqrt{\log n})\},$ $ exp\{O(\sqrt{\log n})\})$-network-decompositions in $exp\{O(\sqrt{\log n})\}$ time. Awerbuch et al. \cite{ABCP96} devised a deterministic algorithm for computing strong $(O(\log n), O(\log n))$-network-decomposition in time $exp\{O(\sqrt{\log n})\}$.  Around the same time Linial and Saks \cite{LS92} devised a randomized algorithm for computing weak $(O(\log n), O(\log n))$-network-decompositions in $O(\log^2 n)$ time. More generally, the algorithm of \cite{LS92} can compute weak $(\lambda, O(n^{1/\lambda} \log n))$-network-decompositions or weak $(O(n^{1/\lambda}), \lambda)$-network-decompositions in time $O(\lambda \cdot n^{1/\lambda} \log n)$.

Observe, however, that all these algorithms \cite{AGLP89,PS95,LS92} require super-logarithmic time, for all choices of parameters. In ICALP'12 the first-named author of the current paper \cite{B12} devised a randomized algorithm for computing strong $(O(1), n^{1/2 + \epsilon})$-network-decomposition in $O(1/\epsilon)$ time. Unlike the algorithms of \cite{AGLP89,PS95,LS92}, the algorithm of \cite{B12} requires {\em constant} time. Its drawback however is its very high parameter $\ell = n^{1/2 + \epsilon}$. In the current paper we alleviate this drawback, and devise a randomized algorithm for computing strong $(exp\{O(\lambda)\}, n^{1/\lambda})$-network-decomposition in time $exp\{O(\lambda)\}$. In other words, the parameter $\lambda$ of our new decompositions can be made $n^{\epsilon}$, for an arbitrarily small constant $\epsilon > 0$, while the running time is still {\em constant} (specifically, $exp\{O(1/\epsilon)\}$).\\ 
{\bf 1.2 Constant-Time Distributed Algorithms}\\
In their seminal paper titled "What can be computed locally?" \cite{NS93} Naor and Stockmeyer posed the following question: which distributed tasks can be solved in {\em constant} time? 
This question is appealing both from theoretical and practical perspectives. From the latter viewpoint it is justified by the emergence of huge networks. The number of vertices in the latter networks may be so large that even mildest dependence of the running time on $n$ may make the algorithm prohibitively slow.

Naor and Stockmeyer themselves \cite{NS93} showed that certain types of weak colorings can be computed in constant time. A major breakthrough in the study of distributed constant time algorithms was achieved though a decade after the paper of \cite{NS93} by Kuhn and Wattenhofer \cite{KW05}. Specifically, Kuhn and Wattenhofer \cite{KW05} showed that an $O(\sqrt{k} \Delta^{1/\sqrt{k}} \log \Delta)$-approximate minimum dominating set\footnote[1]{A subset $U \subseteq V$ in a graph $G = (V,E)$ is a {\em dominating set} if for every $v \in V\setminus U$ there exists $u \in U$, such that $(u,v) \in E$. In the {\em minimum dominating set} (henceforth, MDS) problem the goal is to find a minimum-cardinality dominating set of $G$.}
 can be computed in $O(k)$ randomized time. Here $\Delta = \Delta(G)$ is the maximum degree of the input graph $G$, and $k$ is a positive possibly constant parameter.

An approximation algorithm for another fundamental optimization problem, specifically, for the {\em minimum coloring} problem, was devised by Barenboim \cite{B12} as an application of his aforementioned algorithm for computing network-decompositions. Specifically, it is shown in \cite{B12} that an $O(n^{1/2 + \epsilon})$-approximation for the minimum coloring problem can be computed in $O(1/\epsilon)$ randomized time. (In the minimum coloring problem one wishes to color the vertices of the graph properly with as few colors as possible.) Observe that since approximating the minimum coloring problem up to a factor of $n^{1 - \epsilon}$ is NP-hard \cite{H96,FK98,Z07}, the algorithm of \cite{B12} inevitably has to employ very heavy local computations.

In the current paper we employ our improved network-decomposition procedure to come up with a significantly improved constant-time approximation algorithm for the minimum coloring problem. Specifically, our randomized algorithm provides an $O(n^{\epsilon})$-approximation for the minimum coloring problem  in $exp\{O(1/ \epsilon)\}$ time, for an arbitrarily small constant $\epsilon > 0$.  We also devise a randomized $O(n^{\epsilon})$-approximation algorithm for the {\em minimum $t$-spanner} problem with running time $exp\{O(1/\epsilon)\} + O(t)$, for any arbitarily small constant $\epsilon > 0$. 
(A subgraph $G' = (V,H)$ of a graph $G = (V,E)$, $H \subseteq E$, is a {\em $t$-spanner} of $G$ if for every $u,v \in V$, $\mbox{dist}_{G'}(u,v) \leq t \cdot \mbox{dist}_G(u,v)$. In the {\em minimum $t$-spanner} problem the objective is to compute a $t$-spanner of the input graph $G$  with as few edges as possible.)

Ajtai et al. \cite{AKS80} showed that triangle-free $n$-vertex graphs admit an $O(\sqrt {n} / \sqrt {\log n})$-coloring. This existential bound was shown to be tight by Kim \cite{K95}. We devise a randomized $O(n^{1/2 + \epsilon})$-coloring algorithm for triangle-free graphs with running time $O(1/ \epsilon)$. More generally, we devise a randomized $O(n^{1/k + \epsilon})$-coloring algorithm for graphs of girth greater than $g = 2k, k \geq 2$, with running time $O(1/ \epsilon^2)$. Both results apply for any arbitrarily small $\epsilon > 0$, and, in particular, they show that such graph can be colored with a reasonably small number of colors in constant time. Together with our drastically improved constant-time approximation algorithm for the minimum coloring problem, these results significantly expand the set of distributed problems solvable in constant time. 

Most our algorithms for constructing network-decompositions use only short messages\footnote[1]{The only exceptions are weak network-decompositions from Section \ref{sc:strongdecomp}.} (i.e., messages of size $O(\log n)$ bits), and employ only polynomially-bounded local computations. Although in general graphs our algorithms for $O(n^{1/\epsilon})$-approximate minimum coloring require large messages, our $O(n^{1/2 + \epsilon})$-coloring and $O(n^{1/k + \epsilon})$-coloring algorithms for triangle-free graphs and graphs of large girth employ short messages.  Hence the latter coloring algorithms are suitable to serve as building blocks for various tasks. Despite that the number of colors is superconstant, in many tasks it does not affect the overall running time, so the entire task can be performed very quickly. For example, if the colors are used for frequency assignment or code assignment tasks, the running time will not be affected by the number of colors. Instead, the range of frequencies or codes will be affected. However, this is unavoidable in the worst case, in view of the lower bounds on the chromatic number of triangle free graphs and graph of large girth.\\
{\bf 1.3 The Minimum Dominating Set Problem\\}
The MDS problem is one of the most fundamental classical problems of distributed graph algorithms. Jia et al. \cite{JRS01} devised the first efficient randomized $O(\log \Delta)$-approximation algorithm for the MDS problem with running time $O(\log n \log \Delta)$. Their result was improved and generalized by Kuhn and Wattenhofer \cite{KW05} who devised an $O(k)$-time randomized $O(\sqrt{k} \Delta^{1/\sqrt{k}} \log \Delta)$-approximation algorithm for the problem.

The results of \cite{JRS01,KW05} spectacularly advanced our understanding of the distributed complexity of the MDS problem. However, both these algorithms \cite{JRS01,KW05} are randomized, and no efficient deterministic distributed algorithms with a non-trivial approximation guarantee for general graphs are currently known. Lenzen and Wattenhofer \cite{LW10} devised such algorithms for graphs with bounded arboricity. Below we provide a quote from their paper: \\
\textit{"To the best of our knowledge, the deterministic distributed complexity of MDS approximation on general graphs is more or less a blind spot, as so far neither fast (polylogarithmic time) algorithms nor stronger lower bounds are known"}.

\noindent \ \ \ \ \ In this paper we address this blind spot and devise a deterministic $O(n^{1/k})$-approximation algorithm for the MDS problem with time $O((\log n)^{k-1})$. Similarly to our approximation algorithms for the minimum coloring and the minimum $t$-spanner problems, this algorithm is also a consequence of our algorithms for constructing network-decompositions. However, for the MDS we use a deterministic version of these algorithms, while for the minimum coloring and minimum $t$-spanner problems we use a randomized version. Also, we present a variant of our MDS approximation algorithm that employs only polynomially-bounded local computations, requires $O((\log n)^{k-1})$ time, and provides an $O(n^{1/k} \log \Delta)$ approximation.\\ 
{\bf 1.4 Additional Results \\}
We also use our algorithms for computing network-decompositions for devising algorithms for computing {\em low-intersecting partitions}. Low-intersecting partitions were introduced by Busch et al. \cite{BDRRS12} in a paper on universal Steiner trees.
A {\em low-intersecting $(\alpha, \beta, \gamma)$-partition} ${\cal P}$ of a graph $G$ is the partition of the vertex set $V$ such that: 
(1) Every cluster $C$ in ${\cal P}$ has strong diameter at most $\alpha \cdot \gamma$. \\
(2) For every vertex $v \in V$, a ball $B_{\gamma}(v)$ of radius $\gamma$ around $v$ intersects at most $\beta$ clusters of ${\cal P}$. 

 \ \ \ Busch et al. showed that given a hierarchy of low-intersecting partitions with certain properties (see \cite{BDRRS12} for details) one can construct a universal Steiner tree. (See \cite{BDRRS12} for the definition of universal Steiner tree.) Also, vice versa, given universal Steiner tree they showed that one can construct a low-intersecting partition. They constructed a low-intersecting partition with $\alpha = 4^k, \beta = k \cdot n^{1/k}$, and arbitrary $\gamma$.

We devise a distributed randomized algorithm that constructs low-intersecting $((O(\gamma)^k, n^{1/k}, \gamma)$-partitions in time $(O(\gamma))^k \log^{2/3} n$ in general graphs and in $(O(\gamma))^k \cdot exp \{O(\sqrt{\log \log n})\}$ time in graphs of girth  $g \geq 6$. This algorithm employs only short messages and polynomially-bounded local computations.

Comparing this result with the algorithm of Busch et al. \cite{BDRRS12} we note that the partition of \cite{BDRRS12} has smaller radius. (It is $\gamma \cdot (O(1))^k$ instead of $(O(\gamma))^k$ in our case.) On the other hand, the intersection parameter $\beta$ of our partitions is smaller. (It is $n^{1/k}$ instead of $k \cdot n^{1/k}$.) In particular, the intersection parameter in the construction of \cite{BDRRS12} is always $\Omega(\log n)$, while ours can be as small as one wishes. Finally, and perhaps most importantly, the algorithm of \cite{BDRRS12} is not distributed, and seems inherently sequential.

\noindent {\bf 1.5 Comparison of Our and Previous Techniques\\}
Basically, our algorithms for computing network-decompositions can be viewed as a randomized variant of the deterministic algorithm of Awerbuch et al. \cite{AGLP89}. The algorithm of Awerbuch et al. \cite{AGLP89} computes iteratively ruling sets for subsets of high-degree vertices in a number of supergraphs. These supergraphs are induced by certain graph partitions which are computed during the algorithm. (A subset $U \subseteq V$ of vertices is called an $(\alpha, \beta)$-ruling set if any two distinct vertices $u, u' \in U$ are at distance at least $\alpha$ one from another, and every $v \in V \setminus U$ not in a ruling set has a "ruler" $u \in U$ at distance at most $\beta$ from $v$.)
As a result of the computation the algorithm of \cite{AGLP89} constructs a partition into clusters of diameter at most $\alpha$, such that the supergraph induced by this partition has arboricity at most $\beta$. The algorithm of \cite{AGLP89} then colors this partition with $O(\beta)$ colors in time $O(\beta \log n) \cdot O(\alpha)$. (The running time of the algorithm is $O(\beta \log n)$ when running on an ordinary graph. The running time is multiplied by a factor of $O(\alpha)$, because the coloring algorithm is simulated on a supergraph whose vertices are clusters of diameter $O(\alpha)$.) The fact that the running time in the result of \cite{AGLP89} is (roughly speaking) the product $\alpha \cdot \beta$ of the parameters of the resulting network-decomposition is the reason that Awerbuch et al \cite{AGLP89} made an effort to balance these parameters, and set both of them to be equal to $exp \{ O(\sqrt{\log n \log \log n})\}$. The algorithm of Panconesi and Srinivasan \cite{PS95} is closely related to that of \cite{AGLP89} except that it invokes a sophisticated doubly-recursive scheme for computing ruling sets via network-decompositions, and vice versa.
 This ingenious idea enables \cite{PS95} to balance the parameters and running time better. Specifically, they are all equal to $2^{O(\sqrt{\log n})}$.

Our algorithm is different from \cite{AGLP89,PS95} in two respects. First, we replace a quite slow (it requires $O(\log n)$ time) deterministic procedure for computing ruling sets by a constant-time randomized one. Note that {\em generally} computing $(O(1),O(1))$-ruling sets requires $\Omega(\log^* n)$ time \cite{L92}, but we only need to compute them for {\em high-degree vertices} of certain supergraphs. This can be easily done in randomized constant time. Second, instead of coloring the resulting partition with $O(\beta)$ colors in $O(\beta \log n) \cdot O(\alpha)$ time, we color it in $O(\beta \cdot n^{\epsilon})$ colors in $O(1/\epsilon) \cdot O(\alpha)$ time by a simple randomized procedure, or in $O(\beta^2 \log^{(t)} n)$ colors in $O(t) \cdot O(\alpha)$ time, for a parameter $t>0$, by a deterministic algorithm Arb-Linial \cite{BE08}. Hence the number of colors is somewhat greater than in \cite{AGLP89,PS95}, but the running time is constant.

The algorithm of Linial and Saks \cite{LS92} is inherently different from both \cite{AGLP89,PS95} and from our algorithm. It runs for $O(\log n)$ phases, each of which constructs a collection of clusters of diameter $O(\log n)$ at pairwise distance at least $2$ which covers at least half of all remaining vertices. The running time of the algorithm of \cite{LS92}, similarly to \cite{AGLP89} and \cite{PS95}, is the product of the number of phases and clusters' diameter. Hence the approach of \cite{LS92} appears to be inherently incapable to give rise to a constant time algorithm.

\noindent \ \ \ \ \ Our deterministic variant of the network-decomposition procedure is the basis for our deterministic approximation algorithm for MDS. Our deterministic variant is closer to the algorithm of \cite{AGLP89} than our randomized one. The main difference between our deterministic variant and the algorithm of \cite{AGLP89} is that we use a different much faster coloring procedure for the supergraph induced by the ultimate partition.

\noindent {\bf 1.6 Related Work \\}
Network-decompositions for general graphs were studied in \cite{ABCP96,C93,AP90}. Dubhashi et al. \cite{DMPRS05} used network decompositions for constructing low-stretch dominating sets. Recently, Kutten et al. \cite{KNPR14} extended Linial-Saks network-decompositions to hypergraphs.  Many authors \cite{GV07,KMW05,SW08} studied network-decompositions for graphs with bounded growth. Distributed approximation algorithms is a vivid research area. See, e.g., \cite{N14} and the references therein. Distributed graph coloring is also a very active research area. See a recent monograph \cite{BE13}, and the references therein. Schneider et al. \cite{SEW13} devised a distributed coloring algorithm whose performance depends on the chromatic number of the input graph. However, the algorithm of \cite{SEW13} provides no non-trivial approximation guarantee. To the best of our knowledge there are no known distributed approximation algorithms for the minimum $t$-spanner problem. Efficient distributed algorithms for constructing sparse undirected spanners can be found in \cite{E07,DGPV08}. For centralized approximation algorithms for the minimum $t$-spanner problem, see \cite{KP94,EP05,BBMRY11}.   
\section{Preliminaries}  \label{sc:preliminaries}
\noindent For a subset $V' \subseteq V$, the graph $G(V')$ denotes the subgraph of $G$ induced by $V'$.  The {\em degree} of a vertex $v$ in a graph $G = (V,E)$, denoted {\em $\deg_G(v)$}, is the number of edges incident on $v$.  
A vertex $u$ such that $(u,v) \in E$ is called a {\em neighbor} of $v$ in $G$. The {\em neighborhood} of $v$ in $G$, denoted $\Gamma_G(v)$, is the set of neighbors of $v$ in $G$. If the graph $G$ can be understood from context, then we omit the underscript $_G$. For a vertex $v \in V$, the set $v \cup \Gamma(V)$ is denoted by $\Gamma^+(v)$. For a set $W \subseteq V$, we denote by $\Gamma^+(W)$ the set $W \cup \bigcup_{w \in W} \Gamma(w)$. 
The {\em distance} between a pair of vertices $u,v \in V$, denoted $\mbox{dist}_G(u,v)$, is the length of the shortest path between $u$ and $v$ in $G$. 
The {\em diameter} of $G$ is the maximum distance between a pair of vertices in $G$.
The {\em chromatic number} $\chi(G)$ of a graph $G$ is the minimum number of colors that can be used in a proper coloring of the vertices of $G$. 

\section{Network Decomposition}
\subsection{Procedure Decompose} \label{sc:decompose}
In this section we devise an algorithm for computing an $(O(1),O(n^{\epsilon}))$-network-decomposition in $O(1)$ rounds, for an arbitrarily small constant $\epsilon > 0$.  More generally, our algorithm computes a $(3^k,O(k \cdot n^{2/k} \cdot \log^2 n))$-network-decomposition  $Q$ in $O(3^k \cdot \log^* n)$ rounds, for any positive parameter $k, 1 \leq k  \leq \log n$, along with an $O(k \cdot n^{2/k} \cdot \log^2 n)$-coloring $\varphi$ of the supergraph induced by $Q$. (The $\log^* n$ term can be eliminated from the running time at the expense of increasing the number of colors used by $\varphi$ by a multiplicative factor of $\log^{(t)} n$, for an arbitrarily large constant $t$. We will later show that the multiplicative factor of $k$ in the second parameter of the network decomposition can also be eliminated without affecting other parameters.)  The algorithm is called {\em Procedure Decompose}. 
The procedure runs on some supergraph ${\hat G} = ({\hat V}, {\hat E})$ of the original graph $G$. Each vertex  $C \in {\hat V}$ is a cluster (i.e., a subset of vertices) of the original graph $G = (V,E)$, and different clusters are disjoint. Observe that generally it may happen that $V \neq \cup_{C \in {\hat V}} C$. The procedure accepts as input the supergraph ${\hat G}$, the number of vertices $n$ of $G$, the parameter $k$, and an upper bound $s$ on the number of vertices of the supergraph ${\hat G}$. It also accepts as input two numerical parameters $\epsilon$ and $t$. The parameter $\epsilon > 0$ is a sufficiently small positive constant and $t > 0$ is a sufficiently large integer constant. Initially the supergraph is $G$ itself, with each vertex $v$ forming a singleton cluster $\{ v \}$. Hence initially it holds that $n = s$.
The procedure is invoked recursively. After each invocation the current  supergraph $\hat{G}$ is replaced with a supergraph on fewer vertices, and $s$ is updated accordingly. The parameter $n$, however, remains unchanged throughout the entire execution.) As a result of an execution of Procedure Decompose every vertex $v$ in $\hat{G}$ is assigned a label $label(v)$. The value of $label(v)$ is equal to the color $\varphi(C_v)$ of the cluster $C_v$ of $Q$ which contains $v$.

Procedure Decompose partitions the graph $\hat{G}$ into two vertex-disjoint subgraphs with certain helpful properties. Specifically, one of the subgraphs has a sufficiently small maximum degree that allows us to compute a network decomposition in it directly and efficiently. The other subgraph can be partitioned into a sufficiently small number of clusters with bounded diameter. The latter property is used to construct a supergraph whose vertices are formed from the clusters. Since the number of clusters is sufficiently small, the number of vertices of the supergraph is small as well. Then our algorithm proceeds recursively to compute a network decomposition of the new supergraph, using fresh labels that have not been used yet. The recursion continues for $k$ levels. Then each vertex is assigned the label of the supernode it belongs to. (Supernodes of distinct recursion levels may be nested one inside the other. In this case an inner supernode receives the label of an outer supernode. A vertex of the original graph $G$ receives the (same) label of all supernodes it belongs to. Notice that a vertex belongs to exactly one supernode in each recursion level.) This completes the description of the algorithm. Its pseudocode is provided below. (See Algorithm \ref{proced:decompose}.)

The algorithm employs two auxiliary procedures that we describe in detail in Section \ref{sc:pr}. The procedures succeed with high probability, i.e., with probability $1 - 1/n^c$, for an arbitrarily large constant $c$. The first procedure is called {\em Procedure Dec-Small}. It accepts a graph $G$ with at most $n$ vertices and maximum degree at most $d$. Procedure Dec-Small accepts also as input two numerical parameters, $\epsilon$ and $t$, which are relayed to it from Procedure Decompose. Recall that $\epsilon > 0$ is a sufficiently small constant and $t$ is a sufficiently large integer constant. The procedure computes an $O(\min\{d \cdot n^{\epsilon},d^2\})$-coloring of $G$ in $O(\log^* n)$ time. (The time is $O(1)$ if $d > n^{\epsilon}$. Another variant of this procedure computes an $O(d^2 \log^{(t)} n)$-coloring in $O(t)$ time, for an arbitrarily large positive integer $t$.) Observe that for any integer $p > 0$, a proper $p$-coloring of a graph $G$ is also a $(0,p)$-network-decomposition of $G$. (There are $p$ labels, and each cluster consists of a single vertex. Thus the diameter of the decomposition is $0$.)  Procedure Dec-Small returns a $(0,p)$-network-decomposition $S$ on line 5. It also returns a labeling function $label_S$ for vertices of a subset $A$. (We will soon describe how this subset is obtained.)  The labeling $label_S$ also serves as a proper coloring for the supergraph induced by $S$. 

The second procedure which is invoked by our algorithm is called {\em Procedure Partition}. This randomized procedure accepts as input an $s$-vertex supergraph $\hat{G}=(\hat{V},\hat{E})$ and a parameter $q < \frac{|\hat{V}|}{2c \cdot \log n}$, and partitions $\hat{V}$ into two subsets $A$ and $B$, such that $\hat{G}(A)$ and $\hat{G}(B)$ have the following properties. The subgraph $\hat{G}(A)$ has maximum degree $O(q \log n)$. The subgraph $\hat{G}(B)$ consists of $O(|V|/ q) = O(s/q)$ clusters of diameter at most $2$ with respect to $\hat{G}$. The procedure contracts each such cluster into a supernode. Let ${\cal B}$ denote the resulting set of supernodes and ${\cal G}({\cal B}) = ({\cal B}, {\cal E}({\cal B}))$ the resulting supergraph. Specifically, the vertex set of ${\cal G}({\cal B})$ is ${\cal B}$, and its edge set is  ${\cal E}({\cal B})=\{ (C,C') \ | \ C,C' \in {\cal B}, \  \exists u \in C, u' \in C', \mbox{ such that } (u,u') \in \hat{E}\}.$  Procedure Partition returns the subset $A \subseteq \hat{V}$ and the set of supernodes ${\cal B}$.

 The clusters in $B$ are obtained by computing a dominating set $D$ of $B$ of size $O(|V|/q)$. Each vertex in $D$ becomes a leader of a distinct cluster. Each vertex in $B \setminus D$ selects an arbitrary neighbor in $D$ and joins the cluster of this neighbor. Consequently, in all clusters all vertices are at distance at most $1$ from the leader of their cluster. Hence all clusters have diameter at most $2$. 
Initially, each vertex of $V$ joins the set $D$ with probability $1/q$. Then the set $B$ is formed by the vertices of $D$ and their neighbors. Finally, the set $A$ is formed by the remaining vertices, i.e., $A = V \setminus B$. In this stage the procedure returns the set of nodes $A$ and the set of supernodes ${\cal B}$ which is obtained from $B$, and terminates. This completes the description of Procedure Partition.

\begin{algorithm}[H]
\caption{Procedure Decompose($\hat{G}, n, k, s, \epsilon ,t$)}
\label{proced:decompose}

\begin{algorithmic}[1] 

\IF {$s \leq 2c \cdot n^{1/k} \log n$}

    \STATE return Dec-Small($\hat{G}, n, s, \epsilon, t$) 
 
    /* Compute directly a $(0,O(s^2))$-network-decomposition of $\hat{G}$. (See Section \ref{sc:pr}.) */

\ELSE

    \STATE $(A,{\cal B})$ := Partition($\hat{G}, q := n^{1/k} $)
				
		/* Partition $\hat{G}$ into $A$ and ${\cal B}$. (See Section \ref{sc:pr}.) The maximum degree of $\hat{G}(A)$ is $O(n^{1/k} \log n)$.*/
		
		\STATE $(S, label_S)$ := Dec-Small($G(A), n, n^{1/k} \log n$, $\epsilon$, $t$)
		
		/* Compute directly a $(0,O(n^{2/k} \cdot \log^2 n))$-network-decomposition of $\hat{G}(A)$. (See Section \ref{sc:pr}.) */
		
	  \STATE $(L, label_L) := $ Decompose(${\cal G}({\cal B}), n, k,  \frac{s}{ n^{1/k}} $)
		
		/* A recursive invocation on the supergraph ${\cal G}({\cal B})$ that contains at most $\frac{s}{ n^{1/k}}$ supernodes.  */
		
		\FOR {each vertex $v$ of $\hat{G}$, {\bf in parallel,}}
		
		    \IF {$v \in S$}
				   
					\STATE $label(v) := label_S(v)$
				
				\ELSIF {$v \in L$}
		
		    \STATE $label(v) := label_L(v) + \Lambda$ 
				
				/* $\Lambda = \gamma \cdot \left \lfloor n^{2/k} \cdot \log^2 n \right \rfloor$, where $\gamma$ is a sufficiently large constant to be determined later. */
				
				\ENDIF
				
		/* The labeling function $label$ on $S \cup L$ is defined by: for a cluster $C \in S$ (respectively, $C \in L$) it applies to it the function $label_S()$ (resp., $label_L() + \Lambda$). */		
				
		\ENDFOR
		
		\STATE return $(S \cup L, label)$
		
\ENDIF

\end{algorithmic}
\end{algorithm}

The recursive invocation of Procedure Decompose on line 6 returns a network decomposition $L$ for the supergraph ${\cal G}(B)$. The for-loop (lines 7-13) adds (in parallel) $\Lambda = \gamma \cdot \left \lfloor n^{2/k} \log^2 n \right \rfloor$ to the color of each cluster of the network decomposition $L_0$ of ${\cal G}(B)$, where $\gamma$ is a sufficiently large constant to be determined later. Since the number of colors used in each recursive level is at most $\Lambda$, this loop guarantees that colors used for clusters created on different recursion levels are different. This is because the labeling returned by procedure Dec-Small on line 5 for clusters of $S$ employs the palette $[\Lambda]$ while the labeling computed in lines 10 - 12 for clusters of $L$ employs labels which are greater than $\Lambda$.
The termination condition of the procedure is the case $s = O(n^{1/k} \log n)$, i.e., when the number $s$ of vertices in the supergraph $\hat{G}$ is already small. At this point the maximum degree of $\hat{G}$ is small as well (at most $s-1$), and so coloring the supergraph (by Procedure Dec-Small) results in a sufficiently good network decomposition.

Observe that our main algorithm will invoke the procedure on the original graph $G$. Hence in the first level of the recursion ${\hat G} = G$, and each supernode is actually a node of $G$. In the second recursion level it is executed on the supernodes of nodes of the original graph $G$. In the third level it is executed on supernodes of supernodes, etc. Consequently, starting from the second recursion level supernodes have to be simulated using original nodes of the network. To this end each cluster that forms a supernode selects a leader which is used for simulating the supernodes. Moreover, the leader is used to simulate all nested supernodes to which it belongs. Our supernodes are obtained by at most $k$ levels of nesting. In each level of nesting a supernode is a cluster of diameter at most $2$ in a graph whose nodes are lower-level supernodes. Hence a simulation of a single round on such a supergraph will require up to $3^{k + 1}$ rounds. 

Next we provide several lemmas that will be used for the analysis of the algorithm. We leave the parameters $\epsilon$ and $t$ unspecified in all lemmas in this section, because they have no effect on the analysis. 
 \begin{lem} \label{decomposelevl}
Consider an invocation of Procedure Decompose on the original graph $G = (V,E)$ with parameters $n = |V|$, $k$, and $s = n$, for $1 \leq k \leq \log n$.
The number of recursion levels in the execution of this Procedure (i.e., Decompose(${\hat G} :=G,n,k, s := n$)) is $k$.
\end{lem}
\begin{proof}
In recursion level $i$, $i = 1,2,...$, the parameter $s$ is equal to $n^{1 - (i-1)/k}$. Hence, in recursion level $k$ the parameter $s$ is equal to $n^{1/k}$, and the recursion reaches the termination condition. (See lines 1-2 of Algorithm \ref{proced:decompose}.)
\end{proof} 
\begin{lem} \label{numoflabels}
The number of labels used in the invocation of Procedure Decompose(${\hat G} := G,n,k, s := n$) is $O(k \cdot n^{2/k} \log^2 n)$.
\end{lem}
\begin{proof}
We show that the number of labels is $\gamma \cdot k \cdot \left \lfloor n^{2/k} \log^2 n \right \rfloor$,
where $\gamma$ is a sufficiently large constant.
Specifically, the constant $\gamma$ needs to be larger than the constants hidden by the $O$-notation in comments on lines 2 and 5 of the algorithm. (Recall that line 2 computes a $(0, O(s^2))$-network-decomposition, and line 5 computes a $(0,O(n^{2/k} \log^2 n))$-network-decomposition.
The constant $\gamma$ appears in line 11 of Algorithm \ref{proced:decompose}.)  The proof is by induction on $\ell_i = k - i + 1$, where $i$ is the recursion level. In other words, this is an inverse recursion on the number of recursion levels.
For each index $i \in [k]$, denote by $\hat{G}_i$ the supergraph on which Procedure Decompose is invoked on the $i$th level of the recursion. Note that at this point $s = n^{1 - (i - 1)/k}$. The inductive claim is that the $i$th level invocation of Procedure Decompose (on the supergraph $\hat{G}_i$) employs at most $\ell_i \cdot \gamma \cdot \left \lfloor n^{2/k} \log^2 n \right \rfloor \leq (k - i + 1) \cdot \gamma \cdot n^{2/k} \log^2 n$ labels.\\
{\bf Base ($\ell_i = 1$, i.e., $i = k$)}: In this case $s \leq 2c \cdot n^{1/k} \log n$, the termination condition of the recursion holds, and thus the number of labels used in the decomposition is $O(s^2)$. (See line 2 of Algorithm \ref{proced:decompose}.) By the choice of $\gamma$, the number of labels is at most $\gamma \cdot n^{2/k} \cdot \log^2 n$. \\
{\bf Step}: Suppose that the invocation has returned from level $i + 1$ of the recursion, and it is now at level $i$. By the induction hypothesis, line 6 of Algorithm \ref{proced:decompose} returns a labeling with $ \gamma \cdot \ell_{i+1} \cdot \left \lfloor n^{2/k} \log^2 n \right \rfloor = \gamma \cdot (\ell_i-1) \cdot \left \lfloor n^{2/k} \log^2 n \right \rfloor $ labels. Once line 11 is executed, the number of labels becomes
$\gamma \cdot \ell_i \cdot \left \lfloor n^{2/k} \log^2 n \right \rfloor$. 
This proves the inductive claim.

In the end of recursion level $i = 1$ the algorithm terminates (after returning from all recursive invocations). In this stage it holds that $\ell_1 = k$, and the claim follows.
\end{proof} 
\begin{lem} \label{dec}
Each cluster created by the invocation above has diameter at most $3^{k - 1} - 1$.
\end{lem}
\begin{proof}
We prove by induction on  $\ell = k - i + 1$, where $i$ is the recursion level, that level-$i$ clusters have diameter at most $3^{\ell - 1} -1$. \\
{\bf Base ($\ell = 1$, i.e., $i = k$)}: In this case a $(0,O(s^2))$-network-decomposition is computed directly, and thus the diameter of all clusters in the graph $\hat{G}$ on which it is executed is $0$. (Recall that the argument $\hat{G}$ in the  level-$k$ invocation is a supergraph of the original input graph $\hat{G}$.)\\
{\bf Step:} First, observe that a $(0,O(n^{2/k} \cdot \log^2 n))$-network-decomposition of $\hat{G}(A)$ is computed directly in line 5 of Algorithm \ref{proced:decompose}. Hence $S$ consists of clusters of diameter $0$ (with respect to supernodes of the supergraph $\hat{G}$ of the current recursion level). Next, we analyze the diameter of clusters in $L$. By the induction hypothesis, line 6 of Algorithm \ref{proced:decompose} (i.e., the recursive invocation of Procedure Decompose) returns a network decomposition in which all clusters have diameter at most $3^{\ell-2} - 1$. This is a decomposition of the supergraph ${\cal G}({\cal B})$. Consider a cluster ${\cal C}$ of diameter at most $3^{\ell - 2} - 1$ in ${\cal G}({\cal B})$. 
Let $x,y$ be a pair of vertices of $\hat{G}$ that belong to supernodes in ${\cal C}$. Let $C(x),C(y) \in {\cal C}$ be these two supernodes (clusters), such that $x \in C(x)$, $y \in C(y)$. Since the diameter of ${\cal C}$ in ${\cal G}({\cal B})$ is at most $3^{\ell - 2} - 1$, there exist clusters $C(x) = C_1,C_2,...,C_t = C(y) \in {\cal C}$, such that $t \leq 3^{\ell - 2} -1$, and the following holds. There exist edges $e_1 = (u_1,v_1),e_2 = (u_2,v_2),...,e_{t-1} = (u_{t-1},v_{t-1})$, such that 
for every $i \in [t-2], u_i \in C_i$, and for every $i \in [2, t-1]$, $v_i \in C_{i+1}$. 
(See Figure \ref{clustersvertices} for an illustration.)
By construction, each of the clusters $C_1,C_2,...,C_t$ has diameter at most $2$. Hence for $i \in [t - 2]$, it holds that $\mbox{dist}_{\hat{G}}(v_i,u_{i+1}) \leq 2$. Therefore, 
\begin{eqnarray*}
\mbox{dist}_{\hat{G}}(x,y)  & \leq &  \mbox{dist}_{\hat{G}}(x,u_1) + 1 + \mbox{dist}_{\hat{G}}(v_1,u_2) + 1 + \mbox{dist}_{\hat{G}}(v_2,u_3) + 1 +    \\ & ... & + \mbox{dist}_{\hat{G}}(v_{t - 2},u_{t- 1}) + 1 + \mbox{dist}_{\hat{G}}(v_{t - 1},y)   \ \ \ \ \leq \ \ \ \  2 \cdot t + t - 1 \ \ \ = \ \ \ 3 \cdot t - 1. 
\end{eqnarray*}
Since $t \leq 3^{\ell - 2}$, it follows that $\mbox{dist}_{\hat{G}}(x,y) \leq 3 \cdot 3^{\ell - 2} - 1 = 3^{\ell - 1} - 1$.
Therefore, the diameter of ${\cal C}$ in $\hat{G}$ is at most $3^{\ell - 1} -1$. Hence all clusters in $S \cup L$ have diameter at most $3^{\ell - 1} - 1$ in $\hat{G}$. Since the number of recursion levels is $k$, the claim follows.
\end{proof}
\includegraphics{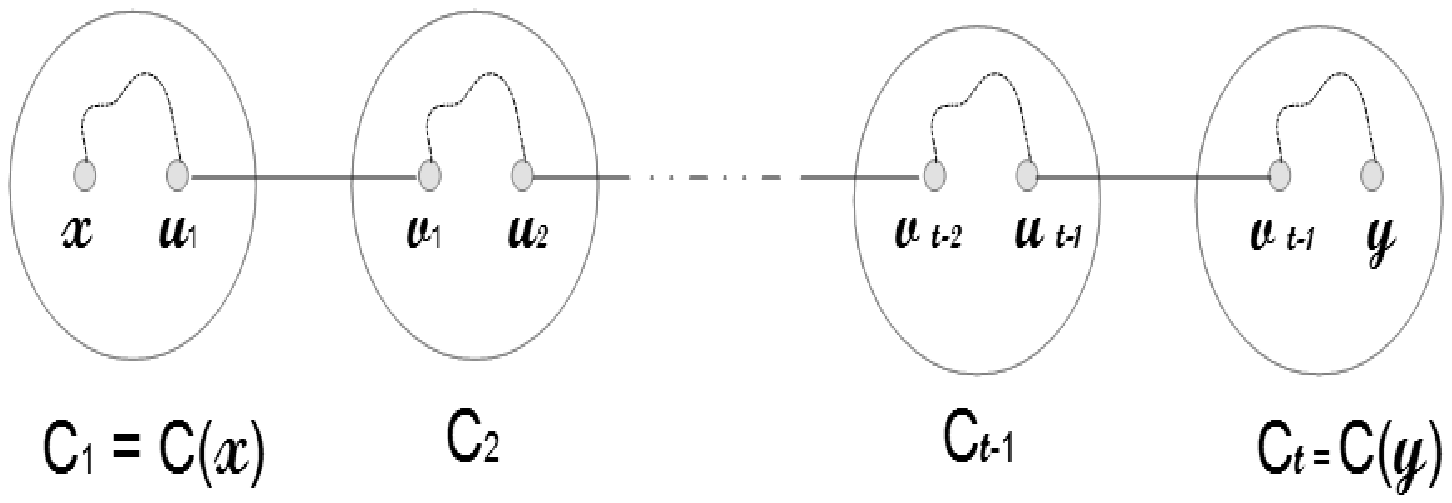}
\begin{fig} \label{clustersvertices}
The clusters $C(x) = C_1,C_2,...,C_t = C(y)$.
\end{fig}
\begin{lem} \label{lemmac}
Suppose that all invocations of auxiliary procedures of Procedure Decompose have succeeded. Then the invocation computes a $(3^{k - 1} - 1, O(k \cdot n^{2/k} \cdot \log^2 n))$-network-decomposition.
\end{lem}
\begin{proof}
Consider a pair of distinct adjacent clusters $C,C' \in S \cup L$. If $C \in S$ and $C' \in L$ then $label(C) = label_S(C) \in [\Lambda]$, while $label(C') = label_L(C') > \Lambda$. Hence $label(C) \neq label(C')$.

If $C,C' \in S$ then since Procedure Dec-Small returns on line 5 a network decomposition with a proper labeling $label_S(\cdot)$, it follows that $label_S(C) \neq label_S(C')$, and so $label(C) \neq label(C')$.

Finally, if $C,C' \in L$ then inductively we conclude that $label_L(C) \neq label_L(C')$, and thus $label(C) \neq label(C')$ too. (The induction base is the recursion level $k$, where the correctness follows from the correctness of Procedure Dec-Small invoked on line 2 of Algorithm \ref{proced:decompose}.)

Hence Procedure Decompose returns a partition $S \cup L$ into clusters of diameter at most $3^{k-1} - 1$ (by Lemma \ref{dec}), and a proper labeling $label(\cdot)$ of this partition. By Lemma \ref{numoflabels}, the number of labels used by the labeling $label(\cdot)$ is $O(k \cdot n^{2/k} \cdot \log^2 n)$. Hence $S \cup L$ is a $(3^{k - 1} - 1, O(k \cdot n^{2/k} \cdot \log^2 n))$-network-decomposition for $G$, and $label(\cdot)$ is a proper labeling for the network decomposition $S \cup L$.
\end{proof}

Recall that the auxiliary procedures Dec-Small and Partition succeed with probability $1 - 1/n^c$, for an arbitrarily large constant $c$. Each of these procedures is invoked at most $k \leq \log n$ times during the execution of Procedure Decompose. Therefore, the probability that all executions of Procedure Dec-Small and Procedure Partition succeed is at least $(1 -1/n^c)^{2 \log n} \approx 1 - \frac{1}{n^c/2 \log n}$. Since $c$ is an arbitrarily large constant, all executions of the auxiliary procedures succeed, with high probability. Hence Procedure Decompose computes a $(3^k,O(k \cdot n^{2/k} \cdot \log^2 n))$-network-decomposition, with high probability.

The next lemma analyzes the running time of the algorithm. 
\begin{lem} \label{dectime}
Let $T_{part}(n,q)$ (respectively, $T_{dec}(n,d)$) denote the running time of Procedure Partition invoked with parameters $n$ and $q$ (resp., Procedure Dec-Small invoked with parameters $n$ and $d$). We will assume that both these running times are monotone non-decreasing in both parameters. 
Then the running time of Procedure Decompose is $O(3^k \cdot (T_{part}(n,n^{1/k}) +    T_{dec}(n,2c \cdot n^{1/k} \log n)))$.
\end{lem}
\begin{proof}
During the execution of Procedure Decompose the Procedure Dec-Small is executed $k$ times, and Procedure Partition is executed $k-1$ times. For $i = 1,2,...,k-1$, in recursion level $i$ both procedures are executed on supergraphs whose supernodes constitute subgraphs of diameter at most $3^i$ of the original graph. Thus, the number of rounds required in level $i$ is the product of the number of steps required to execute the procedure on the supergraph and the maximum diameter of supernodes. This running time is at most $(T_{part}(n,n^{1/k}) + T_{dec}(n,n^{1/k} \log n) + O(1)) \cdot 3^i$. The running time of the last recursion level $k$ in which the termination condition holds is $T_{dec}(n,2c \cdot n^{1/k} \log n) \cdot 3^k$. Therefore, the overall running time is \\
$ O(\sum_{i = 1}^k 3^i \cdot (T_{part}(n,n^{1/k}) + T_{dec}(2c \cdot n^{1/k} \log n))) = O(3^k \cdot (T_{part}(n,n^{1/k}) + T_{dec}(n,2c \cdot n^{1/k} \log n)))$.
\end{proof}
Procedure Dec-Small and Procedure Partition are provided and analyzed in Section \ref{sc:pr}. Next we state the main results obtained by plugging these procedures into Procedure Decompose. See Section \ref{sc:pr} for the proofs.

\begin{thm} \label{dlarge}
For any parameter $k, 1 \leq k \leq \log n$, Procedure Decompose computes a $(3^k,O(k \cdot n^{2/k} \cdot \log^2 n))$-network-decomposition along with the corresponding $O(k \cdot n^{2/k} \cdot \log^2 n)$-labeling function in time $O(3^k \cdot \log^* n)$, with high probability. Alternatively, one can also have the second parameter equal to $O(k \cdot n^{2/k} \log n)$ and the running time $O(3^k \cdot k)$.
\end{thm}


It follows that, an $(O(1), n^{\delta})$-network-decomposition of an arbitrary $n$-vertex graph along with a proper $n^{\delta}$-labeling for it can be computed by a randomized algorithm, in $O(1)$ time, with high probability. See Section \ref{sc:refine}.


\subsection{Procedure Dec-Small and Procedure Partition} \label{sc:pr}
We start with the description of Procedure Dec-Small. This procedure accepts a graph $G$ with at most $n$ vertices and maximum degree at most $d$, and computes an $O(\min\{d \cdot n^{\epsilon},d^2\})$-coloring of $G$, where $\epsilon$ is a fixed arbitrarily small positive constant. In other words, if $d \leq n^{\epsilon}$ then an $O(d^2)$-coloring is computed, and otherwise an $O(d \cdot n^{\epsilon})$-coloring is computed. For computing an $O(d^2)$-coloring, Procedure Dec-Small employs the deterministic algorithm of Linial \cite{L92} that computes an $O(\Delta^2)$-coloring of graphs with maximum degree $\Delta$ within $O(\log^* n)$ time. For computing an $O(d \cdot n^{\epsilon})$-coloring, Procedure Dec-Small employs the randomized algorithm of Barenboim \cite{B12} that computes, with high probability, an $O(\Delta \cdot n^{\epsilon} )$-coloring in $O(1/\epsilon)$ time, for an arbitrarily small $\epsilon > 0$.  We henceforth refer to this algorithm as {\em Procedure Random-Color}. This completes the description of Procedure Dec-Small. Its pseudocode is provided below. 

For completeness, we provide a high-level description of the algorithm of Linial \cite{L92} and the algorithm of Barenboim \cite{B12}. The algorithm of Linial \cite{L92} starts with a legal $n$-coloring of the input graph obtained from the IDs of the vertices. It proceeds in phases, each of which reduces the number of colors while preserving the legality of the coloring. In each round the number of colors is reduced from $p$ to $O(\Delta^2 \log p)$, where $p$ is the number of colors in the beginning of a round. (Initially $p = n$.) In the last round the number of colors is reduced from $O(\Delta \cdot \mbox{polylog}(\Delta))$ to $O(\Delta^2)$. Each phase requires just a single round, and the overall running time of the algorithm of Linial is $O(\log^* n)$.
(It is actually $\log^* n + O(1)$, but this precision is immaterial for our purposes.)

Observe also that one can run Linial's algorithm for just $t$ rounds, for some positive integer parameter $t$, and obtain an $O(\Delta^2 \cdot \log^{(t)} n)$-coloring.
For a single phase of Linial's algorithm it employs $\Delta$-union free set systems from the paper by Erdos, Frankel and Furedi \cite{EFF85}. A family ${\cal F}$ of sets over a given ground-set $X$ is said to be $\Delta$-union-free if for every $\Delta + 1$ sets $S_0,S_1,...,S_{\Delta} \in {\cal F}$, it holds that $S_0 \nsubseteq \cup_{i = 1}^{\Delta} S_i$. Erdos et al. showed that for any positive integers $p$ and $\Delta$, $p \geq \Delta + 1$, there exists a $\Delta$-union-free family ${\cal F}$ of $p$ subsets over a ground-set $X$ of size $|X| = O(\Delta^2 \log n)$.

Let $\varphi$ be a proper $p$-coloring of $G$ in the beginning of a phase of Linial's algorithm. The algorithm associates a set $S_c$ from ${\cal F}$ with each color $c$ of $\varphi$. Every vertex $v$ that runs the algorithm computes a new color $c' \in S_{\varphi(v)} \setminus \cup_{u \in \Gamma(v)} S_{\varphi(u)}$. Such a color exists since ${\cal F}$ is a $\Delta$-union-free family, and $\varphi(u) \neq \varphi(v)$ for every $u \in \Gamma(v)$.) The vertex $v$ sets its new color $\varphi'(v)$ by $\varphi'(v) = c'$. Since $c' \in X$, $|X| = O(\Delta^2 \log n)$, it follows that $\varphi'(\cdot)$ is an $O(\Delta^2 \log n)$-coloring. Also, consider a pair of neighbors $v$ and $u$. Observe that $\varphi'(v) \in S_{\varphi(v)} \setminus S_{\varphi(u)}$, while $\varphi'(u) \in S_{\varphi(u)}$. Hence $\varphi'(v) \neq \varphi'(u)$, and thus $\varphi'$ is a proper coloring.
See \cite{L92} or \cite{BE13} Chapter 3.10 for more details.

The algorithm of Barenboim \cite{B12} (Procedure Random-Color) proceeds in phases as well, however, each phase consists of a randomized procedure. In this procedure each vertex either succeeds in selecting a final color, or fails and continues to the next phase. Each vertex selects a color from the range $[\Delta \cdot n^{\epsilon}]$ uniformly at random, where $\epsilon > 0$ is an arbitrarily small constant. A vertex succeeds if and only if it selects a color that has not been selected by any of its neighbors (either in the current round, or as a final color in a previous round). Otherwise, it fails, discards its color, and continues to the next phase. Hence, the probability that a vertex $v$ fails
 to select a color that is different from the colors of its neighbors is at most $1/n^{\epsilon}$. 
If we run this procedure for $\left \lceil c/ \epsilon \right \rceil$ rounds, for a sufficiently large constant $c$, then the probability that a given vertex $v$ fails on all these rounds is at most $1/n^c$. Hence, by union bound, after $\left \lceil c/ \epsilon \right \rceil$ rounds all vertices succeed with probability at least $1 - 1/n^{c-1}$, i.e., with high probability.
\begin{algorithm}[H]
\caption{Procedure Dec-Small($G, n, d, \epsilon, t$)}
\label{proced:dec-small}

\begin{algorithmic}[1] 

\IF  {$d \leq n^{\epsilon}$}
  
  \STATE compute an $O(d^2)$-coloring of $G$ using the algorithm of Linial \cite{L92}
	
	/* alternatively, one can compute here an $O(d^2 \log^{(t)} n)$-coloring in $O(t)$ time */
  
\ELSE

  \STATE compute an $O(d \cdot n^{\epsilon})$-coloring of $G$ using Procedure Random-Color
  
\ENDIF

\end{algorithmic}
\end{algorithm}

Observe that if $d^2 \leq d \cdot n^{\epsilon}$ then $d \leq n^{\epsilon}$ and an $O(d^2)$-coloring is computed. Otherwise, $d^2 > d \cdot n^{\epsilon}$, and an $O(d \cdot n^{\epsilon})$-coloring is computed. Therefore, the number of colors is $O(\min\{d^2, d \cdot n^{\epsilon}\})$.
Recall also that the running time of the algorithm of Linial \cite{L92} is $O(\log^* n)$, and the running time of Procedure Random-Color is $O(1)$. Thus we obtain the following lemma.
\begin{lem} \label{dsmall}
Procedure Dec-Small invoked on a graph $G$ with maximum degree at most $d$ computes, with high probability, an $O(\min\{d^2, d \cdot n^{\epsilon}\})$-coloring, which is a $(0,O(\min\{d^2, d \cdot n^{\epsilon}\}))$-network-decomposition. If $d \leq n^{\epsilon}$, the running time of Procedure Dec-Small is $O(\log^* n)$. Otherwise, it is $O(1)$.
\end{lem}
Another variant of the Procedure Dec-Small checks if $d \leq \frac{n^{\epsilon}}{\log^{(t}) n}$, and if it is the case it invokes the $t$-round version of Linial's algorithm. Otherwise it invokes line 4 of Algorithm \ref{proced:dec-small} (i.e., the algorithm from \cite{B12}).
This modified procedure always requires constant time. (Assuming that $t = O(1)$.) If $d \leq \frac{n^{\epsilon}}{\log^{(t)} n}$, it computes an $O(d^2 \log^{(t)}n)$-coloring.
Otherwise (in this case $d^2 \log^{(t)} n > d \cdot n^{\epsilon}$) it computes a $(d \cdot n^{\epsilon})$-coloring. To summarize:
\begin{lem}
A modified variant of Procedure Dec-Small computes an $O(\min \{d^2 \log^{(t)} n, d \cdot n^{\epsilon} \})$-coloring in $O(t)$ time, with high probability. In particular, the running time is constant if $t = O(1)$.
\end{lem}

Next, we describe Procedure Partition. Procedure Partition accepts as input a graph $G=(V,E)$ and a positive parameter $q$, and partitions $V$ into two subsets $A$ and $B$, such that $G(A)$ and $G(B)$ have the following properties. The subgraph $G(A)$ has maximum degree $O(q \log n)$. The subgraph $G(B)$ consists of $O(|V|/ q)$ clusters of diameter at most $2$.
The procedure contracts the clusters of $B$ into supernodes, which form the supergraph ${\cal G}({\cal B})$.  The clusters in $B$ are obtained by computing a dominating set $D$ of $B$ of size $O(|V|/q)$. Each vertex in $D$ becomes a leader of a distinct cluster. Each vertex in $B \setminus D$ selects an arbitrary neighbor in $D$ and joins the cluster of this neighbor. Consequently, in all clusters all vertices are at distance at most $1$ from the leader of their cluster. Hence all clusters have diameter at most $2$. 

Initially, each vertex of $V$ joins the set $D$ with probability $1/q$. Then the set $B$ is formed by the vertices of $D$ and their neighbors. Finally, the set $A$ is formed by the remaining vertices, i.e., $A = V \setminus B$. In this stage the procedure returns the set of nodes $A$ and the set of supernodes ${\cal B}$ which is obtained from $B$, and terminates. This completes the description of the procedure. Its pseudocode is provided below.
\begin{algorithm}[H]
\caption{Procedure Partition($G, q$)}
\label{proced:partitin}
An algorithm for each vertex $v \in V$.
\begin{algorithmic}[1] 

\STATE $v$ joins $D$ with probability $1/q$ and informs its neighbors

\IF {$v$ has joined $D$ or a neighbor of $v$ has joined $D$}

     \STATE $v$ joins $B$
     
\ELSE
   
      \STATE $v$ joins $A$
      
\ENDIF

\IF { $v \in D$}

      \STATE $v$ initializes a singleton cluster $C_v$, becomes the leader of $C_v$, and sends $Id(v)$ to all neighbors in $B$

\ENDIF

\IF {$v \in B$ and $v$ receives at least one message from a leader of a cluster in $D$}

      \STATE $v$ joins a cluster of an arbitrary neighbor in $D$
      
\ENDIF

\STATE ${\cal B} :=$ the set of supernodes obtained by contracting all clusters $C_v \subseteq B$

\STATE return $(A,{\cal B})$

\end{algorithmic}
\end{algorithm}
Note that the sets $A$ and ${\cal B}$ are returned in a distributed manner. In other words, each vertex knows whether it belongs to $A$ or to $B$. If it belongs to $B$, then it knows the identity of the leader of its cluster. The leaders of the clusters represent the supernodes formed by the clusters. Thus, when we say that a supernode performs some action, it is actually performed by the leader of the cluster that forms the supernode. (For nested supernodes the operations are performed by the leaders of the innermost clusters, which are vertices in the original input graph $G$.)

In the next lemmas we prove that Algorithm \ref{proced:partitin} computes a partition with the properties described above.
\begin{lem}
Suppose that Procedure Partition is invoked on a graph $G = (V,E)$ and a positive parameter $q$. Then the subgraph $G(A)$ induced by the set $A$ which is returned by the procedure has maximum degree $O(q \cdot \log n)$, with high probability.
\end{lem}
\begin{proof}
Consider a vertex $v \in V$ such that $v$ has at least $c \cdot q \cdot \ln n$ neighbors in $G$, for a sufficiently large constant $c$. Denote $\delta = \deg_G(v)$. Let $y_1,y_2,...,y_{\delta}$ denote the neighbors of $v$ in $G$, and let $y_0 = v$. 
The probability that none of these neighbors join $D$ is
$$Pr(y_1 \notin D,..., y_{\delta} \notin D) = \Pi_{i = 1}^{\delta} \Pr(y_i \notin D) = \Pi_{i = 1}^{\delta} (1 - 1/q) = (1 - 1/q)^{\delta} \leq (1 - 1/q)^{c \cdot q \cdot \ln n} \leq 1/n^c.$$
Hence, by union bound, the probability that at least one vertex $v$ with at least $c \cdot q \cdot \ln n$ neighbors does not have a neighbor in $D$ is at most $1/n^{c - 1}$. Hence with probability at least $1 - 1 /n^{c - 1}$, all high-degree vertices (vertices with $\deg_G(v) \geq c \cdot q \cdot \ln n$) end up in $B$. Hence, with high probability, the maximum degree of a vertex in $A = V \setminus B$ is $O(q \log n)$.
\end{proof}
The next lemma analyzes the number of supernodes and their diameters.

\begin{lem} \label{partitnscnd}
Suppose that Procedure Partition is invoked on a graph $G = (V,E)$ and a parameter $q  < \frac{|V|}{2c \cdot \log n}$, for some constant $c > 1$. Then the set ${\cal B}$ returned by the procedure has the following properties. With high probability, ${\cal B}$ consists of $O(|V|/q)$ supernodes. All supernodes of ${\cal B}$ are clusters of diameter at most $2$ in $G$.
\end{lem}
\begin{proof}
Recall that the set ${\cal B}$ is created by contracting the clusters of $B$ into supernodes. First, we prove that all clusters of $B$ have diameter at most $2$. Let $C$ be a cluster of $B$. Let $u,v \in C$ be any pair of vertices in the cluster. Then either one of these vertices is the leader of the cluster and $\mbox{dist}(u,v) = 1$ or both $u$ and $v$ are connected to the same leader, and so $\mbox{dist}(u,v) \leq 2$. (Since $u$ and $v$ belong to $B$ they must have a leader neighbor, unless they are leaders themselves. Since $u$ and $v$ belong to the same cluster, and there is exactly one leader in each cluster, if $u$ and $v$ are not the leaders, they are connected to the same leader.) Next, we prove that ${\cal B}$ consists of $O(|V|/q)$ supernodes. Note that the number of supernodes in ${\cal B}$ is equal to the number of vertices in $D$, since each vertex in $D$ becomes a leader of a cluster (i.e., of a supernode). Let $X$ denote a random variable that counts the number of vertices in $D$. Since each vertex in $V$ joins $D$ with probability $1/q$ independently of other vertices, it holds that $\Expect(X) = |V|/q$. By Chernoff bound for upper tails (see, e.g., \cite{MU05}, Chapter 4),\\
$ Pr[X > 2 |V|/q] = Pr[X > 2 \Expect(X)] \leq (e/4)^{\Expect(X)} = (e/4)^{|V|/q} \leq (e/4)^{2c \cdot \log n} \leq 1/n^c.$
\end{proof}
Finally, note that each line of Procedure Partition is either performed locally, or involves  sending messages to neighbors. The latter requires one time unit. Therefore, the running time of Procedure Partition is $O(1)$.
\begin{lem} \label{prptimecns}
Procedure Partition requires $O(1)$ time.
\end{lem}
Combining Lemmas \ref {dec} - \ref{dectime} with Lemmas \ref{dsmall} - \ref{partitnscnd} imply the following results.
\begin{thm} 
For any parameter $k, 1 \leq k \leq \log n$, Procedure Decompose computes a $(3^k,O(k \cdot n^{2/k} \cdot \log^2 n))$-network-decomposition along with the corresponding $O(k \cdot n^{2/k} \cdot \log^2 n)$-labeling function in time $O(3^k \cdot \log^* n)$, with high probability.
\end{thm}
Consider now a variant of Algorithm \ref{proced:decompose} (Procedure Decompose) in which in Procedure Dec-Small we always invoke Procedure Random-Color with a parameter $\epsilon$. (As opposed to Algorithm \ref{proced:dec-small} where we do it only when $d \leq n^{\epsilon}$.) Also, in Algorithm \ref{proced:decompose} we now set $\Lambda \leftarrow (2 \cdot c \cdot n^{1/k} \cdot \log n) \cdot n^{\epsilon}$.

Then, by the previous argument, this modified variant of Procedure Decompose computes a $(3^k, O(k \cdot n^{1/k + \epsilon} \log n))$-network-decomposition along with a legal $O(k \cdot n^{1/k + \epsilon} \log n)$-labeling function in time $O(3^k/ \epsilon)$. By substituting $\epsilon = 1/k$ we conclude:
\begin{thm} 
A $(3^k, O(k \cdot n^{2/k} \log n))$-network-decomposition along with the appropriate proper (with respect to this decomposition) $O(k \cdot n^{2/k} \log n)$-labeling can be computed in $O(k \cdot 3^k)$ time, with high probability.
\end{thm}
 
In particular, by setting $k$ to be an arbitrarily large constant we obtain an $(O(1), n^{\delta})$-network-decomposition along with a proper $n^{\delta}$-coloring for it in randomized constant time, for an arbitrarily small constant $\delta > 0$.
\begin{col}
An $(O(1), n^{\delta})$-network-decomposition of an arbitrary $n$-vertex graph along with a proper $n^{\delta}$-labeling for it can be computed by a randomized algorithm, in $O(1)$ time, with high probability.
\end{col}
A yet another variant of Procedure Decompose (Algorithm \ref{proced:decompose}) is obtained if in Proc Dec-Small we always invoke the $t$-round variant of Linial's algorithm \cite{L92}, for some positive integer parameter $t$. (Again we do it regardless of the value of $d$.) Also, for this variant we set $\Lambda = \gamma \cdot n^{2/k} \log^2 n \log^{(t)} n$, where $\gamma > 0$ is a sufficiently large constant. (Specifically, the $t$-round variant of Linial's algorithm computes an $O(\Delta^2 \log^{(t)} n)$-coloring of the input $n$-vertex graph with maximum degree $\Delta$. The constant $\gamma$ should be larger than the constant hidden by the $O$-notation in $O(\Delta \log^{(t)} n)$.)

The resulting algorithm computes a $(3^k, O(k \cdot n^{2/k} \cdot \log^2 n \cdot \log^{(t)} n))$-network-decomposition with a proper $O(k \cdot n^{2/k} \cdot \log^2 n \cdot \log^{(t)})$-labeling for it, in $O(3^k \cdot t)$ time.
\begin{col}
For any $n$-vertex graph $G$ and parameters $k = 1,2,...; \ t = 1,2,...; \ \epsilon > 0$, one can compute a $(3^k, O(k \cdot n^{1/k + \epsilon} \cdot \log n))$-network-decomposition (respectively, $(3^k, O(k \cdot n^{2/k} \cdot \log^2 n \cdot \log^{(t)} n))$-network-decomposition) with an appropriate labeling function in $O(3^k/\epsilon)$ (resp., $O(3^k \cdot t)$) randomized time, 
\end{col}





\section{Refining the Algorithm} \label{sc:refine}
In this section we argue that one can save a factor of $k$ in the number of labels, and compute a $(3^k, O(n^{1/k + \epsilon} \log n))$-network-decomposition (and a $(3^k, O(n^{2/k} \log n \log^{(t)} n))$-network-decomposition) with an appropriate labeling in $O(3^k/\epsilon)$ (resp., in $O(3^k \cdot t)$) time. While this improvement is negligible when $k$ is small, it becomes significant when $k$ is superconstant. We remark, however, that in the context of the current paper we are mainly interested in the regime of small $k$.

To describe this improvement we need the notions of arboricity and $H$-partition. (We refer the reader to \cite{BE08} and \cite{BE13} for a more elaborate discussion on this topic.)

The {\em arboricity} $a(G)$ of a graph $G = (V,E)$ is the minimum number $t$ of edge-disjoint forests $F_1,F_2,...,F_t$, such that $E = \cup_{i = 1}^t F_i$. An {\em $H$-partition} $(H_1,H_2,...,H_{\ell})$ of $G = (V,E)$ with degree at most $A$, for some number $A$, is a partition of the vertex set $V$ of $G$ into vertex disjoint subsets $V = \cup_{i = 1}^{\ell} H_i$, $H_i \cap H_j = \emptyset$ for every pair of distinct indices $i \neq j$, $i,j \in [\ell]$, such that for every index $i \in [\ell]$ and every vertex $v \in H_i$, the number of neighbors $\deg(v, \cup_{j=i}^{\ell} H_j)$ that $v$ has in $H$-sets $H_j$ with an index $j \geq i$ is at most $A$.

Consider again Procedure Decompose. (See Algorithm \ref{proced:decompose}.) For $i = 1,2,...,k$, let $\hat{G}_i$ denote the supergraph on which the procedure is invoked in the $i$th level of recursion. In particular, $\hat{G}_1 = G$ is the original graph. Also, in all levels $i \leq k -1$ the procedure enters lines 4 - 15, and in the last level $i = k$ it enters the termination condition (line 2). In the former case Procedure Decompose invokes Procedure Dec-Small (in line 5), which returns the collection $S$ of clusters. (It also returns the labeling function that is immaterial for the current discussion.) For $i = 1,2,...,k-1$, let $S_i$ denote the set of clusters returned by Procedure Dec-Small on line 5 of the $i$th level recursive invocation of Procedure Decompose. Finally, in level $k$ of the recursion Procedure Dec-Small is invoked in line 2. Denote by $S_k$ the decomposition that it returns.
\begin{lem} \label{sets}
$\cup_{i=1}^k S_i$ is the network decomposition that Procedure Decompose returns (in line 14 of the first level recursive invocation). Moreover, for any index $i$, $k \geq i \geq 1$, $\cup_{j = i}^k S_j$ is the network decomposition that the $i$th level recursive invocation of Procedure Decompose returns.
\end{lem}
\begin{proof}
The proof is by induction on $i$. \\
{\bf Base ($i = k$)}: In this case Procedure Decompose returns the output $S_k$ of an invocation of Procedure Dec-Small (on line 2 of Algorithm \ref{proced:decompose}).\\
{\bf Step}: Consider some $1 \leq i < k$. The $i$th level recursive invocation returns $S \cup L$ in line 14. Recall that $S = S_i$ is a network decomposition for $\hat{G}(A)$ computed in line 5 of Algorithm \ref{proced:decompose}. (In all levels except the first one $S$ is actually equals to $A$. In the first level $S = \{\{v\} \ | \ v \in A\}$.)
Also, $L$ (computed by the $(i + 1)$st level recursive invocation of Procedure Decompose; see line 6 of Algorithm \ref{proced:decompose}) is a network decomposition for ${\cal G}({\cal B})$. By induction hypothesis the latter is  $\cup_{j = i + 1}^k S_j$. Hence $S \cup L = \cup_{j= i}^k S_j$, as required.
\end{proof}
In the next lemma we show that $(S_1,S_2,...,S_k)$ is an $H$-partition with relatively small degree of the supergraph ${\cal G}(\cup_{i=1}^k S_i)$ induced by the network decomposition $\cup_{i=1}^k S_i$.
\begin{lem} \label{partitions}
$(S_1,S_2,...,S_k)$ is an $H$-partition with degree $O(n^{1/k} \log n)$ of the supergraph ${\cal G}(\cup_{i=1}^k S_i)$.
\end{lem}
\begin{proof}
Again, by an induction on $i$, $k \geq i \geq 1$, we show that $(S_i,S_{i+1},...,S_k)$ is an $H$-partition of ${\cal G}(\cup_{i=1}^k S_i)$ with maximum degree $O(n^{1/k} \log n)$.\\
{\bf Base ($i = k$)}: In this case we need to show that the maximum degree in ${\cal G}(S_k)$ is $O(n^{1/k} \log n)$. By the termination condition of Algorithm \ref{proced:decompose} (line 1), $|S_k| = O(n^{1/k} \log n)$, and thus the same upper bound applies to its maximum degree.\\
{\bf Step}: For some $i$, $1 \leq i \leq k$, we argue that for any cluster $C \in S_i$ its degree in ${\cal G}(\cup_{j = i}^k S_j)$ is $O(n^{1/k} \log n)$. By Lemma \ref{sets}, $\cup_{j=i}^k S_j$ is the network decomposition for ${\cal G}(\cup_{j = i}^k S_j)$ that the $i$th level recursive invocation of Procedure Decompose returns. By construction, $S_i$ is the set of clusters with degree at most $O(n^{1/k} \log n)$ in the supergraph $\hat{G}_i$. Since clusters of $\cup_{j = i + 1}^k S_j$ are obtained by merging clusters of $\hat{G}_i$, it follows that the degree of $C$ in $\cup_{j = i}^k S_j$ is no greater that its degree in $\hat{G}_i$, i.e., at most $O(n^{1/k} \log n)$.
\end{proof}
To recap, Lemma \ref{partitions} shows that in addition to computing a network decomposition $Q$ of its input graph $G$, Procedure Decompose also computes a low-degree $H$-partition of the induced supergraph ${\cal G}(Q)$. (Here Q = $\cup_{i = 1}^k S_i$, and the $H$-partition is $(S_1,S_2,...,S_k)$. The degree of the partition is $O(n^{1/k} \log n)$.)

For the variant of Procedure Decompose that we describe in this section we do not actually need to explicitly compute the labeling function during the execution of the procedure. As a result Procedure Dec-Small can be greatly simplified. Specifically, if it is invoked on a supergraph $\hat{G} = (\hat{V}, \hat{E})$, then it returns $\hat{V}$ as its output partition. If it is invoked on a subgraph $G(U) = (U,E(U))$ of the original graph $G$, then it returns a partition of $U$ into singleton clusters, i.e., $\{\{u\} \ | \ u \in U \}$. Observe that in this simplified form Procedure Dec-Small requires $O(1)$ time. As a result the overall running time of Procedure Decompose becomes $O(3^k)$ rather than $O(3^k/ \epsilon)$ or $O(3^k \cdot t)$.

Next, we utilize the $H$-partition $(S_1,S_2,...,S_k)$ of ${\cal G}(Q)$ for computing an $O(n^{2/k} \log^2 n \log^{(t)} n)$-coloring of ${\cal G} (Q)$ in $O(3^k \cdot t)$ time, or alternatively, an $O(n^{1/k + \epsilon} \log n)$-coloring of ${\hat G} (Q)$ in $O(3^k / \epsilon)$ time. Such colorings can be viewed as labelings of the network decomposition $\cup_{i=1}^k S_i = Q$. (For every vertex $v \in V$, its label will be equal to the color of the cluster $C_v \in Q$ that contains it.)

To simplify presentation, consider an $n$-vertex graph $G' = (V',E')$ and an $H$-partition $(S_1,S_2,...,S_k)$ for $G'$ with degree $A = O(n^{1/k} \log n)$. We will argue that $G'$ can be efficiently colored. To implement this coloring in a supergraph ${\cal G} (Q)$, we will need to multiply the running time by the maximum diameter of a cluster in $Q$, i.e., by $O(3^k)$. We start with arguing that $G'$ can be colored in $O(A^2 \log^{(t)} n)$ colors in $O(t)$ time, by a deterministic algorithm. (This algorithm is closely related to Algorithm Arb-Linial from \cite{BE08}, based on Linial's algorithm \cite{L92}. The current algorithm is however more general than Algorithm Arb-Linial.) The algorithm starts by orienting all edges $(u,v)$ in the following way: let $i_u$ (respectively, $i_v$) be the index of the set $S_{i_u}$ (resp., $S_{i_v}$) which contains $u$ (resp., $v$). If $i_u < i_v$ then the edge is oriented towards $v$. If the opposite holds than it is oriented towards $u$. If $i_u =i_v$ then the edge is oriented towards the endpoint with a greater $Id$.

Observe that under this orientation each vertex $v$ has at most $A$ outgoing edges incident on it. The opposite endpoints of these edges will be referred to as the {\em parents} of $v$. Let $\Pi(v)$ denote the set of parents of $v$.

Let $\varphi$ be a proper $p$-coloring of $G'$. We argue that a legal $O(A^2 \log p)$-coloring $\varphi'$ of $G'$ can be computed within one single round. To this end we again employ an $A$-union-free family ${\cal F}$ of $p$ sets (due to \cite{EFF85}, see also Section \ref{sc:pr} of this paper). Each color class $c$ of $\varphi$ is associated with a set $X_c \in {\cal F}$. A vertex $v$ computes a color $\varphi'(v)$ which belongs to $X_{\varphi(v)} \setminus \cup_{u \in \Pi(v)} X_{\varphi(u)}$. Such a color necessarily exists, because ${\cal F}$ is an $A$-union-free family. Also, for an edge $(v,u) \in E'$, suppose without loss of generality that $u \in \Pi(v)$. Then $\varphi'(v) \in X_{\varphi(v)} \setminus X_{\varphi(u)}$, while $\varphi'(u) \in X_{\varphi(u)}$, and so $\varphi'(v) \neq \varphi'(u)$. By \cite{EFF85}, a family ${\cal F}$ over a ground-set of size $O(A^2 \log p)$ exists (and can be efficiently constructed). Thus, $\varphi'$ is a proper $O(A^2 \log p)$-coloring. By repeating this recoloring step for $t$ times, we obtain an $O(A^2 \log^{(t)} n)$-coloring in $O(t)$ rounds. (We start with an initial $n$-coloring of $G'$. Specifically, each vertex uses its $Id$ as its initial color.)

\begin{col}
An $O(n^{2/k} \log^2 n \log^{(t)} n)$-coloring of ${\cal G} (Q)$ can be computed in $O(3^k \cdot t)$ time, for any $t = 1,2,...$.
\end{col}
Observe that this argument shows in fact that the arboricity of ${\cal G} (Q)$ is $O(A) = O(n^{1/k} \log n)$, and thus $Q$ is a $(3^{k - 1} - 1, O(n^{1/k} \log n))$-network-decomposition. We summarize these results in the following corollary.
\begin{col} \label{partdec}
Procedure Decompose, invoked on an $n$-vertex graph $G = (V,E)$ with a parameter $k = 1,2,...$, computes a $(3^{k - 1} -1, O(n^{1/k} \log n))$-network-decomposition $Q$ and an $H$-partition $(S_1,...,S_k)$ of degree $O(n^{1/k} \log n)$ and length $k$ for ${\cal G} (Q)$ in $O(3^k)$ randomized time, with high probability. Moreover, for a parameter $t = 1,2,...$, one can compute in $O(3^k \cdot t )$ time an $O(n^{2/k} \log^2 n \log^{(t)}n)$-labeling for $Q$. In particular, by setting $t = \log^* n$ one can get here time $O(3^k \log^* n)$ and labeling with $O(n^{2/k} \log^2 n)$ labels.
\end{col}
Note that the $O(A^2 \log^{(t)} n)$-coloring algorithm for $G'$ that was described above does not require the fact that the $H$-partition $(S_1,S_2,...,S_k)$ has small number of sets. Next we show that this $H$-partition can be used in a more explicit way to compute an $O(A \cdot n^{\epsilon})$-coloring of $G'$ in $O(k / \epsilon)$ time.

First, every vertex $v$ of $S_k$ tosses a color $\varphi(v)$ uniformly at random from the palette $[A \cdot n^{\epsilon}]$. It checks if its color is different from the colors of all its neighbors in $S_k$. If it is the case, it finalizes its color. Otherwise, it tosses its color from the same palette again. The process is repeated for $\left \lceil c / \epsilon \right \rceil$ rounds, for a sufficiently large constant $c$. As we have already seen (see Lemma \ref{dsmall} and the discussion preceding it), in $O(1/ \epsilon)$ rounds we will obtain a legal $O(A \cdot n^{\epsilon})$-coloring $\varphi_k$ for $S_k$, with high probability. (Recall that the maximum degree in $S_k$ is at most A.) Define also $\hat{\varphi}_k = \varphi_k$. 

Suppose that we have already computed an $O(A \cdot n^{\epsilon})$-coloring $\hat{\varphi}_i$ for $G'(\cup_{j=i}^k S_j)$, for some $i$, $2 \leq i \leq k$. Next we show how to extend this coloring into an $O(A \cdot n^{\epsilon})$-coloring $\hat{\varphi}_{i-1}$ for $G'(\cup_{j = i-1}^k S_j)$. To this end every vertex $v \in S_{i-1}$ tosses a color from $[A \cdot n^{\epsilon}]$ uniformly at random, and checks if its color is different from the colors (either tossed on this round, or finalized colors) of its neighbors in $\cup_{j = i -1}^k S_j$. If it is different from them, then $v$ finalizes its color. Otherwise, it continues to the next round. The entire process continues for $O(1/ \epsilon)$ rounds.

The key observation required for the analysis is that $v \in S_{i -1}$ has at most $A$ neighbors in $\cup_{j = i- 1}^k S_j$ Thus, a legal $(A \cdot n^{\epsilon})$-coloring $\varphi_{i-1}$ for $S_{i - 1}$ will be computed, with high probability, within additional $O(1/\epsilon)$ rounds. It is then combined in a trivial way with the coloring $\hat{\varphi}_i$ for $\cup_{j = i}^k S_j$ to obtain the $(A \cdot n^{\epsilon})$-coloring $\hat{\varphi}_{i-1}$ for $\cup_{j = i -1}^k S_j$.
\begin{thm} \label{arbc}
Given an $H$-partition $(S_1,S_2,...,S_k)$ with degree at most $A$ for an $n$-vertex graph $G' = (V',E')$, and a parameter $\epsilon > 0$, an $(A \cdot n^{\epsilon})$-coloring of $G'$ can be computed in $O(k / \epsilon)$ rounds.
\end{thm}
By invoking this algorithm on the network decomposition $Q$ we obtain:
\begin{col}
Using a $(3^{k-1} - 1, O(n^{1/k} \log n))$-network-decomposition $Q$ of the input graph $G$ and an $H$-partition $(S_1,S_2,...,S_k)$ for $Q$ with degree $O(n^{1/k} \log n)$, one can compute an $O(n^{1/k + \epsilon} \log n)$-labeling for $Q$ within $O(3^k \cdot k/ \epsilon)$ randomized time.
\end{col}
By substituting $\epsilon = 1/k$ we get:
\begin{col} \label{improvement}
A $(3^{k-1} - 1, O(n^{1/k} \log n))$-network-decomposition $Q$ along with an $O(n^{2/k} \log n)$-labeling for it can be computed in $O(3^k \cdot k^2)$ randomized time.
\end{col}

\section{Decompositions with a smaller number of labels} \label{sc:betterlabels}
When $k$ is small the logarithmic factor in the number of labels ($O(n^{2/k} \log n)$) of the network decomposition $Q$ from Corollary \ref{improvement} is almost negligible. However, for large $k$ (e.g., $k = \Omega(\log n)$) this logarithmic factor becomes dominant. In this section we describe a modification of our algorithm that produces $(exp\{O(k)\}, O(n^{1/k}))$-network-decomposition in $exp \{O(k)\} \cdot \log^{2/3} n$ time. (For graphs of girth at least 6 the running time of this algorithm is even better, specifically $exp\{O(k)\} \cdot exp \{O(\sqrt{\log \log n } ) \}$.)
Observe that for $k = \Omega(\log \log n)$, the overhead factor of $\log^{2/3} n$ can be swallowed by the $O$-notation in $exp\{O(k)\}$. This version of our algorithm is closely related to the deterministic algorithm of Awerbuch et al. \cite{AGLP89}; in fact, our algorithms in this section can be viewed as a randomized version of their algorithm. Their deterministic algorithm requires time $(\log n)^{O(k)}$, and so we essentially show here that their algorithm can be made faster by means of randomization.

The difference between the new variant of our algorithm (which we introduce here; we will refer to it as Procedure RS-Decompose) and the original version of our algorithm (described in Section \ref{sc:decompose}) is a different algorithm for Procedure Partition. (See Algorithm \ref{proced:partitin}.) The new variant of Procedure Partition which we will next describe will be called Procedure RS-Partition. (RS stands for the acronym of "ruling set".) 

In a graph $G = (V,E)$ for a vertex set $U \subseteq V$ and positive integer parameters $r, \delta$ a subset $W \subseteq U$ is called an {\em $(r, \delta)$-ruling set} for $U$ if the following two conditions hold: \\
(a) Every pair of distinct vertices $w,w' \in W$ satisfy $\mbox{dist}_G(w,w') \geq \delta$.\\
(b) For every vertex $u \in U$ there exists a "ruling vertex" (also called "ruler") $w \in W$ such that $\mbox{dist}_G(w,u) \leq r$.\\

Observe that an MIS is a $(2,1)$-ruling set. In the description of Procedure RS-Partition we will assume that we have an efficient distributed subroutine for computing $(r, \delta)$-ruling sets for $r = 3$ and $\delta = O(1)$. We will later elaborate on this subroutine.
Procedure RS-Partition starts with computing a $(3, \delta)$-ruling set $W$ for the set $U = \{u \in V \ | \ \deg(u) \geq q \}$ of high-degree vertices of $G$. (Recall that $q$ is an input parameter of Procedure RS-Partition.) Then every vertex $w \in W$ sends an exploration message to distance $\delta$. Every vertex $v \in V$ that receives an exploration message from two distinct rulers $w',w'' \in W$ assigns himself to the ruler $w$ which is closer to it. (Ties are broken in an arbitrary but consistent manner by comparing rulers' identities.)

As a result of these explorations clusters $\{ C_w \ | \ w \in W\}$ are formed. Observe that these clusters all have strong radius at most $\delta$, and that every $u \in U$ (i.e., every high-degree vertex) is assigned to some cluster. (This collection of clusters is often called a {\em ruling forest}. See, e.g., \cite{AGLP89}.) Procedure RS-Partition now forms the set ${\cal B}$ of supernodes by contracting these clusters $G_w$, exactly as in line 13 of Algorithm \ref{proced:partitin}. Further, it creates the set $A$ by setting $A \leftarrow V \setminus (\cup_{w \in W} C_w)$, i.e., every vertex $v$ which is not clustered is assigned to $A$. Observe that for every $v \in A$, it holds that $\deg(v) \leq q$. Finally, Procedure RS-Partition returns the pair $(A, {\cal B})$, exactly as in line 14 of Algorithm \ref{proced:partitin}.
\begin{lem} \label{rprulin}
Suppose that Procedure RS-Partition is invoked on a graph $G = (V,E)$ and a positive parameter $q$. Suppose further that it uses a subroutine for computing a $(3, \delta)$-ruling set, for a positive integer parameter $\delta$. The the subgraph $G(A)$ has maximum degree smaller $q$. Moreover, ${\cal B}$ consists of at most $|V|/q$ supernodes, each of which is a cluster of strong diameter at most $2 \delta$.
\end{lem}
\begin{proof}
All the assertions of the lemma were already argued in the preceding discussion, except for the claim that $|{\cal B}| \leq |V|/q$. We next show this claim. Recall that every supernode of ${\cal B}$ originated from a cluster $C_w$, $w \in W$, where $W$ is a $(3, \delta)$-ruling set for the set $U$ of vertices with degree at least $q$. Hence for two distinct clusters $C_w,C_{w'}$ from the collection ${\cal B} = \{C_w \ | \ w \in W \}$, it holds that $\mbox{dist}_G(w,w') \geq 3$, and $\deg(w), \deg(w') \geq q$. All (immediate) neighbors of $w$ (respectively, $w'$) are assigned to the cluster $C_w$ (resp., $C_{w'}$), and these sets of neighbors are disjoint. Hence $|C_w| \geq q$ for every $w \in W$, and $|{\cal B}| \leq |V|/q $.
\end{proof}

We now use Procedure RS-Partition instead of Procedure Partition within Procedure RS-Decompose. The diameter of clusters in the modified procedure becomes $(2\delta + 1)^k$ instead of $3^k$, but the factor $\log n$ is shaved from the bound on arboricity.
 (This is because the bound on $\deg(A)$ for A returned by Procedure RS-Partition is $q$ instead of $O(q \cdot \log n)$.
Hence as a result we obtain a $((2\delta + 1)^k, O(n^{1/k}))$-network-decomposition $Q$ and an $H$-partition $(S_1,S_2,...,S_k)$ of degree $O(n^{1/k})$ of length $k$ for ${\cal G}(Q)$. (See Corollary \ref{partdec} for a comparison.)

To analyze the running time we need to specify the black-box procedure for computing a $(3, \delta)$-ruling set $W$ for the set $U$ of high degree vertices. Barenboim et al. \cite{BEPS12} (based on \cite{KP12} and \cite{BKP14}) showed that $(2,2)$-ruling sets can be computed in $O(\log^{2/3} \Delta) + exp\{O(\sqrt{\log \log n } )\}$ randomized time in general graphs, and that $(2,3)$-ruling sets can be computed in graphs with girth at least $6$ in just $exp\{O(\sqrt { \log \log n} ) \}$ time.
By running their routine in $G^2$ we guarantee that any two distinct vertices $w,w' \in W$ are at distance at least $2$ in $G^2$, i.e., at distance at least $3$ in $G$. On the other hand, the domination parameter grows by a factor of $2$, i.e., we obtain a $(3,4)$-ruling set in $O(\log^{2/3} \Delta) + exp\{O(\sqrt{\log \log n } )\}$ time in general graphs, 
and a $(3,6)$-ruling set in $exp\{O(\sqrt{\log \log n})\}$ time in graphs of girth at least $6$. 
Hence the running time of Procedure RS-Partition becomes now $O(\log^{2/3} n)$ for general graphs and $exp\{O(\sqrt{\log \log n})\}$ for graphs of girth at least $6$, instead of the running time of $O(1)$ for Procedure Partition. (See Lemma \ref{prptimecns}.) Hence by Lemma \ref{dectime}, the overall running time of Procedure RS-Decompose becomes $O((2 \delta + 1)^k \cdot \log^{2/3} n)$ for general graphs, and $O((2 \delta + 1)^k \cdot exp\{O(\sqrt{\log \log n})\}$ for graphs of girth at least $6$.
In the former case $\delta = 4$, while in the latter it is $6$. To summarize, we have proved the following theorem.
\begin{thm} \label{rsdecom}
Procedure RS-Decompose invoked on an $n$-vertex graph $G = (V,E)$ with a parameter $k = 1,2,...,$ computes an $((O(1))^k, n^{1/k})$-network-decomposition $Q$ and an $H$-partition $(S_1,S_2,...,S_k)$ of degree $A = O(n^{1/k})$ of length $k$ for ${\cal G}(Q)$ in $(O(1))^k \log^{2/3} n$ randomized time for general graphs, and in $(O(1))^k \cdot exp \{O(\sqrt{\log \log n }) \}$ time in graphs of girth at least $6$.
\end{thm}
See Corollary \ref{partdec} for the comparison between the result here and the result that we have for the original variant of our algorithm.

Also in a way analogous to Corollary \ref{improvement}, Theorem \ref{rsdecom} implies that we can also compute a labeling for the network-decomposition $Q$. The time required to compute an $O(A \cdot n^{1/k})$-coloring for ${\cal G}(Q)$ given an $H$-partition as above is, by Theorem \ref{arbc}, at most $O(Diam(Q) \cdot k^2) = (O(1))^k = exp\{O(k)\}$. The number of labels (colors) is $O(A \cdot n^{1/k}) = O(n^{2/k})$. We summarize the properties of the network-decomposition $Q$ in the next corollary.
\begin{col}
An $(exp \{O(k) \}, n^{1/k})$-network-decomposition $Q$ along with an $O(n^{2/k})$-labeling for it can be computed in $exp\{O(k)\} \cdot \log^{2/3} n$ (respectively, $exp\{O(k)\} \cdot O(\sqrt{\log \log n})\}$ ) randomized time in general graphs (resp., in graphs of girth at least $6$).
\end{col}
Observe that randomization was used by the modified variant of Procedure Decompose only for computing a ruling set and for computing the labeling. There is a deterministic algorithm for computing $(3, O(\log n))$-ruling sets in $O(\log n)$ time due to \cite{AGLP89}. If we plug it in the above algorithm the diameter of $Q$ grows from $(O(1))^k$ to $O(\log n)^{k-1})$, and consequently, the running grows to $O((\log n)^{k-1})$ as well. (The most time-consuming step involves computing a $(3, O(\log n))$-ruling set in the last phase of the algorithm, i.e., in a supergraph in which each cluster has diameter $(O(\log n))^{k - 2}$. This requires $(O(\log n))^{k-1}$ time.)
Hence we obtain the following result, which is a generalization of the network decomposition of \cite{AGLP89}. (They arrived to the same result with $k = \sqrt{\log n \log \log n}$, i.e., they obtained an $(exp \{ O (\sqrt{\log n \log \log n } ) \}, exp \{ O(\sqrt{\log n \log \log n} ) \}$)-network-decomposition.)
\begin{col} \label{polylgdecom}
An $(((O(\log n))^{k - 1}, n^{1/k})$-network-decomposition $Q$ along with  an $H$-partition $(S_1,S_2,...,S_k)$ of degree $A = O(n^{1/k})$ of length $k$ for ${\cal G}(Q)$ can be computed in deterministic $(O(\log n))^{k - 1}$ time in general graphs.
\end{col}
This also gives rise to a construction of $O((\log n)^{k-1})$-spanner with $O(n^{1+1/k})$ edges, in deterministic $O((\log n)^{k-1})$ time, in the CONGEST model. This is achieved by adding one edge for every pair of adjacent clusters of the decomposition of Corollary \ref{polylgdecom}. By setting $k =  \frac{\log n}{c \log \log n}$, for a constant $c>1$, one can get $O(n^{1/c})$ time and $O(n \cdot \log^{c} n)$ edges. In particular, this results in a sparse skeleton (with $n \cdot \mbox{polylog}(n)$ edges), in time $O(n^{\epsilon})$, for an arbitrarily small constant $\epsilon > 0$, in the deterministic CONGEST model.

Using the $H$-decomposition of ${\cal G}(Q)$ from Corollary \ref{polylgdecom} an $O(A^2) = O(n^{2/k})$-labeling for it (i.e., $O(A^2)$-coloring) for ${\cal G}(Q)$ can be computed by Algorithm Arb-Linial within additional $O(Diam(Q) \cdot \log^* n) = (O(\log n))^{k-1} \cdot \log^* n$ deterministic time. (This is another point in which this deterministic routine is different from that of \cite{AGLP89}. To compute the coloring Awerbuch et al. \cite{AGLP89} used here $O(Diam(Q) \cdot A \cdot \log n) = O(n^{1/k} \log^k n)$ time, but the number of colors was only $O(n^{1/k})$ instead of $O(n^{2/k})$. Since we insist on having a deterministic polylogarithmic time, this modification is crucial.)
\begin{col} \label{polylogdecom}
For any positive integer $k$, an $((O(\log n))^{k-1}, n^{1/k})$-network-decomposition $Q$ along with an $O(n^{2/k})$-labeling of it can be computed in $(O(\log n))^{k - 1} \cdot \log^* n$ deterministic time.
\end{col}
By running a $t$-round version of Algorithm Arb-Linial, for some positive integer constant $t$, one can also have here running time $O((\log n)^{k-1} \cdot t) = O((\log n)^{k-1})$, but the number of colors (labels) becomes $O(n^{2/k} \log^{(t)} n)$.

\section{Applications}
\subsection{An Approximation Algorithm for the Coloring Problem}
The results described in the previous sections (Theroem \ref{dlarge}; see also Corollary \ref{improvement}) imply an approximation algorithm for the optimization variant of the coloring problem. 
A distributed approximation algorithm for the graph coloring problem (based on an $(O(1),O(n^{1/2 + \epsilon}))$-network decomposition) was given in \cite{B12}. We describe here a generalization of that algorithm which works with any network-decomposition.
The algorithm starts by computing a $(3^k - 1, O(n^{1/k} \log n))$-network-decomposition $Q$ with an $O(n^{2/k} \log n)$-labeling $label(\cdot)$ for it. See Corollary \ref{improvement}. Then in each cluster $C$ the entire induced subgraph $G(C)$ is collected into the leader vertex $v_C$ of $C$. The leader vertex $v_C$ computes locally the optimum coloring $\varphi_C$ for $C$. Finally, $v_C$ broadcasts (a table representation of $\varphi_C$) to all vertices of $C$. Each vertex $u$ that receives this broadcast computes its final color $\psi(u)$ by $\psi(u) = \langle \varphi_C(u), label(u) \rangle$. The running time of this algorithm is the sum of the time required to compute the decomposition $Q$ (i.e., $O(3^k \cdot k^2)$) with the time required for the computation of the colorings $\varphi_C$. The latter is dominated by the diameter of $Q$, times a small constant. The overall running time is therefore $O(3^k \cdot k^2)$.

The next lemma shows that the coloring $\psi$ provides an $O(n^{2/k} \log n)$-approximation to the optimal coloring for $G$. 
\begin{lem} \label{appcol}
$\psi$ is a proper $O(n^{2/k} \log n \cdot \chi(G))$-coloring.
\end{lem}
\begin{proof}
Consider an edge $(u,w) \in E$. If $u,w \in C$, for some cluster $C \in Q$, then $\varphi_C(u) \neq \varphi_C(w)$, and so $\psi(u) \neq \psi(w)$. Otherwise, let $C_u$ (respectively, $C_w$) be the cluster that contains $u$ (resp., $w$), and $C_u \neq C_w$. The clusters $C_u$ and $C_w$ are adjacent in ${\cal G} (Q)$, and thus $label(C_u) \neq label(C_w)$. Hence $label(u) \neq label(w)$, and so $\psi(u) \neq \psi(w)$. 

Note also that $\chi(G(C)) \leq \chi(G)$, for every vertex subset $C \subseteq V$. The coloring $\psi$ employs $\max_{C \in Q} \{\chi(G(C)) \} \cdot n^{2/k} \cdot \log n$ colors, i.e., $O(\chi(G) \cdot n^{2/k} \cdot \log n)$.
\end{proof}
We proved the following theorem:
\begin{thm} \label{mincol}
For any $n$-vertex graph $G = (V,E)$ and an integer parameter $k = 1,2,...$, an $O(n^{2/k} \log n)$-approximation of the optimal coloring for $G$ can be computed in $O(3^k \cdot k^2)$ time.
\end{thm}
In particular, by setting the parameter $k$ to be an arbitrarily large constant we can get a distributed $O(n^{\epsilon})$-approximation algorithm for the coloring problem with a {\em constant} running time, for an arbitrarily small constant $\epsilon > 0$. (The running time is $O(3^{\left \lceil 1/\epsilon \right \rceil} \cdot \frac{1}{\epsilon^2})$.)
This greatly improves the current state-of-the-art constant-time distributed approximation algorithm for the coloring problem due to \cite{B12}, which provides an approximation guarantee of $O(n^{1/2 + \epsilon})$. On the other hand, the dependence of the running time on $\epsilon$ is only $O(1/ \epsilon)$ in the result of \cite{B12}.

Note that the algorithm in Theorem \ref{mincol} requires very heavy (exponential in $n$) local computations and large messages. The heavy computations are inevitable, because unless $NP = P$, the coloring problem cannot be approximated up to a ratio of $n^{1 - \epsilon}$, for any constant $\epsilon > 0$ \cite{H96,FK98,Z07}.
\subsection{Coloring Triangle-Free Graphs and Graphs with Large Girth}
A result of Ajtai et al \cite{AKS80} shows that triangle-free $n$-vertex graphs $G$ admit an $O(\sqrt{n} / \sqrt{\log n})$-coloring. (This existential bound was shown to be tight by Kim \cite{K95}.) Here we show that one can construct an $O(n^{1/2 + \epsilon})$-coloring of triangle-free graphs in distributed randomized $O(1/\epsilon)$ time. Moreover, unlike our algorithm from the previous section, this algorithm uses only {\em short} messages and does not rely on heavy local computations.

The algorithm starts with invoking the algorithm from Corollary \ref{partdec} on its input $n$-vertex graph $G = (V,E)$ with the parameter $k = 2$. We obtain a $(2, O(n^{1/2} \log n))$-network decomposition $Q$ in $O(1)$ time. Moreover, the algorithm also constructs an $H$-partition $(S_1,S_2)$ of the vertex set $V$ of $G$ into two sets. The degree of this $H$-partition is $A = O(n^{1/2} \log n)$. The clusters in $S_1$ are singleton clusters. (Each such a cluster $C \in S_1$ contains a single vertex $v \in C$ such that $\deg(v) \leq  A$.) Each cluster $C \in S_2$ is a star rooted at a center vertex $r \in C$. Also, since the graph is triangle-free, neighbors of $r$ are not connected via edges one with another.

Centers of clusters of $S_2$ now toss a color for their cluster from $[A \cdot n^{\epsilon}]$. If a color tossed by the root $r$ of $C$ is different from the colors of clusters incident on $C$ in the supergraph ${\cal G} (Q)$, then $r$ stops. Otherwise it continues. Overall, as we have seen, after $O(1/ \epsilon)$ rounds, clusters of $S_2$ will be colored in $O(n^{1/2 + \epsilon} \log n)$ colors. (The communication between centers of adjacent clusters can be executed efficiently using short messages.
This requires some care. The root $r$ of each cluster informs all vertices of $C$ of its choice of color. Then each vertex of $C$ (including $r$) sends the root's color $c(r)$ over inter-cluster edges incident on them. Then every vertex $v \in C$ checks if one of its neighboring clusters chose a color $c$ equal to $c(r)$. If it is the case, then it informs $r$. In this case $r$ abandons its color (and informs all vertices of $C$ about it), and continues to the next round of the randomized coloring procedure.) Then clusters of $S_1$ toss colors for them from the same range. Since each cluster of $S_1$ has only $O(A)$ neighbors in $S_1 \cup S_2$, the coloring will be computed within additional $O(1/ \epsilon)$ rounds.
Finally within each cluster $C \in S_2$ actually two colors are used. (One for the center, and another for its neighbors.) Hence the overall number of colors is at most $2 \cdot A \cdot n^{\epsilon} = O(n^{1/2 + \epsilon} \log n)$. The factor $\log n$ can be swallowed by slightly increasing the $\epsilon$ in the exponent. To summarize:
\begin{thm} \label{colortrianglefree}
An $O(n^{1/2+ \epsilon})$-coloring of triangle-free $n$-vertex graph can be computed in $O(1/ \epsilon)$ distributed randomized time, using short messages and polynomially-bounded local computations.
\end{thm}
This result extends also to graphs with large girth. Specifically, consider a graph $G$ with girth greater than $g$, for some integer $g = 2k, k\geq 2$. The arboricity of $G$ is at most $n^{1/k}$. (See, e.g., \cite{B04}. Theorem 3.7.) By \cite{BE08}, an $H$-partition $S_1,S_2,...,S_{\ell}$, $\ell = O(1/\epsilon)$, of $G$ with degree $A = n^{1/k + \epsilon}$ can be computed in constant time, for an arbitrarily small $\epsilon > 0$. Hence, by Theorem \ref{arbc}, an $A \cdot n^{\epsilon}$-coloring of $G$ can be computed in $O(\ell/\epsilon) = O(1)$ time. By scaling $\epsilon$ we obtain the following result.
\begin{thm} \label{colorlargegirth}
For a graph $G$ with girth greater than $g = 2k, k \geq 2$, and an arbitrarily small constant $\epsilon > 0$, an $n^{1/k + \epsilon}$-coloring can be computed in constant distributed randomized time (specifically, $O(1/ \epsilon^2)$),  using short messages and polynomially-bounded local computations.
\end{thm}

Note that the algorithm from Theorem \ref{colorlargegirth} does not employ a network decomposition. Observe also that for $k = 2$ (i.e., girth greater than $4$) the numbers of colors in Theorems \ref{colortrianglefree} and \ref{colorlargegirth} are the same, and both are existentially tight up to a slack factor of $n^{\epsilon}$. On the other hand, their proofs are different. However, Theorem \ref{colortrianglefree} applies for $g > 3$ too, while Theorem \ref{colorlargegirth} applies only for $g > 4$. So the result of Theorem \ref{colorlargegirth} is mainly of interest for $k \geq 3$ (i.e., $g \geq 6$).

\subsection{Separated Decompositions} \label{sc:strongdecomp}
For the sake of some applications we need a stronger notion of network decompositions, called a {\em $\sigma$-separated $(\alpha, \beta)$-network-decomposition}, for positive parameters $\sigma$, $\alpha$, and $\beta$ \cite{ABCP96}. An $(\alpha, \beta)$-network-decomposition $Q$ of a graph $G = (V,E)$ is called {\em $\sigma$-separated} if the clusters of $Q$ can be $\beta$-colored in such a way that every pair of clusters $C, C' \in Q$ which are colored by the same color are at distance at least $\sigma$ from one another, i.e., $\mbox{dist}_G(C,C') \geq \sigma$. Observe that an ordinary network decomposition is $2$-separated.

It is very easy to convert any procedure that constructs an ordinary ($2$-separated) $(\alpha, \beta)$-network-decomposition into a procedure that constructs a {\em weak} $\sigma$-separated $(\alpha \cdot (\sigma - 1), \beta)$-network-decomposition, for any parameter $\sigma \geq 3$. (See Section \ref{sc:preliminaries} for the definition of weak decomposition.) Specifically, one just executes the procedure for computing an ordinary $(\alpha, \beta)$-network-decomposition on the graph $G^{\sigma - 1} = (V,E^{\sigma - 1}), E^{\sigma - 1} = \{ (u,v) \ | \ u,v \in V, \  \mbox{dist}_G(u,v) \leq \sigma - 1 \}$. As a result one obtains a partition $Q$ of $G^{\sigma - 1}$ such that each cluster $C \in Q$ has diameter at most $\alpha$ in $G^{\sigma - 1}$, and thus weak diameter at most  $(\sigma - 1) \cdot \alpha$ in $G$.
Also, for any pair $C,C'$ of distinct clusters in $Q$ which are colored by the same color, the distance between them in $G^{\sigma - 1}$ is at least $2$, and so the distance between them in $G$ is at least $\sigma$. Hence $Q$ is a weak $\sigma$-separated $(\alpha \cdot (\sigma - 1), \beta)$-network-decomposition of $G$. Simulating a distributed algorithm for $G^{\sigma - 1}$ in $G$ increases the running time by a factor of $\sigma - 1$. (Here we assume that message size is unbounded.) Therefore, Corollary \ref{partdec} implies the following result.
\begin{col} \label{sdecpr}
For a pair of positive integer parameters $\sigma \geq 2, k \geq 2$, a $\sigma$-separated weak $((3^{k - 1} - 1) \cdot \sigma, O(n^{1/k} \log n))$-network-decomposition $Q$ and an $H$ partition $(S_1,S_2,...,S_k)$ of length $k$ and degree $O(n^{1/k} \log n)$ for ${\cal G}(Q)$ can be computed in randomized $O(3^k \cdot \sigma)$ time, with high probability. Moreover, for an integer parameter $t = 1,2,...$, one can compute an $O(n^{2/k} \log^2 n \log^{(t)} n)$-labeling for $Q$ in $O(3^k \cdot \sigma \cdot t)$ time.
\end{col}
We remark that this simple approach for converting network-decompositions into weak separated ones is not new. It was used, e.g., by Dubhashi et al. \cite{DMPRS05}.

Next we show that our algorithm for constructing ordinary $(3^{k-1} -1, O(n^{1/k} \log n))$-network-decompositions can be adapted to compute {\em strong} $\sigma$-separated $((2 \sigma - 1)^{k-1} - 1, O(n^{1/k} \log n))$-network-decomposition in randomized time $O((2 \sigma)^k)$, for an arbitrary integer parameter $\sigma \geq 2$.

In what follows we describe Procedure Sep-Decompose which generalizes Procedure Decompose (Algorithm \ref{proced:decompose}). It accepts as input all the parameters of Procedure Decompose, and also the separation parameter $\sigma$. Consider again Procedure Decompose (Algorithm \ref{proced:decompose}). The termination condition of the procedure (lines 1-2, the case when the size $s$ is small, i.e., $s = O(n^{1/k} \log n)$) stays unchanged. In the general case (the "else" case of the procedure, lines 3-15) Procedure Decompose starts with invoking Procedure Partition, which decomposes the input graph $\hat{G}$ into $A$ and ${\cal B}$. In the original procedure the subgraph $G(A)$ induced by $A$ has a small maximum degree (at most $O(q \log n)$, where $q = O(n^{1/k})$ is an input parameter of Procedure Partition.) The generalized variant of the procedure (Procedure Sep-Decompose) invokes instead a generalized variant of Procedure Partition, called {\em Procedure Sep-Partition}. The latter procedure accepts as input all the parameters of Procedure Partition, but also the separation parameter $\sigma$. It also decomposes the input graph $\hat{G}$ into $A$ and ${\cal B}$, but $A$ has the property that $\hat{G}^{(\sigma - 1)}(A)$ has maximum degree $O(q \log n) = O(n^{1/k} \log n)$, i.e., for every vertex $v \in A$, there are at most $O(q \log n)$ other vertices of $\hat{G}$ at distance at most $\sigma - 1$ from $v$. (The distance is with respect to $\hat{G}$.) Similarly to Procedure Partition, in Procedure Sep-Partition too the set ${\cal B}$ is a collection of at most $s/ n^{1/k}$ clusters of small diameter in $\hat{G}$. However, the diameter grows from $3$ in the case of Procedure partition, to $2 \sigma - 1$ in Procedure Sep-Partition.

Then Procedure Sep-Decomposition invokes Procedure Dec-Small. (See line 5 of Algorithm \ref{proced:decompose}.)
Procedure Dec-Small converts every vertex $C \in A$ into a separate cluster. (If $\hat{G}$ is the original graph $G$ then every vertex $v \in A \subseteq V$ gives rise to a cluster $\{v\}$. Otherwise $\hat{G}$ is a supergraph of the original graph $G$, and a vertex $C \in A$ is a cluster of $G$.) The resulting set of clusters is denoted by $S$. Procedure Dec-Small also returns a labeling for clusters of $S$, but similarly to the case of Section \ref{sc:refine}, this labeling is immaterial for our discussion.

On line 6 of algorithm \ref{proced:decompose} Procedure Sep-Decompose invokes itself recursively on the supergraph ${\cal G}({\cal B})$ induced by the set ${\cal B}$ of clusters. The rest of the procedure stays unchanged.

At this point we are interested in a version of Procedure Sep-Decompose which only computes a separated network-decomposition without a labeling function for it. (See the beginning of Section \ref{sc:refine}.) To recap, this procedure returns a network-decomposition $Q = \cup_{i = 1}^k S_i$, where $(S_1,S_2,...,S_k)$ is an $H$-partition of the supergraph ${\cal G}(Q)$ induced by this decomposition. (See Lemmas \ref{sets} and \ref{partitions}.) Moreover, it is easy to verify that decompositions $Q$ produced by Procedure Sep-Decompose satisfy a stronger property than decompositions produced by Procedure Decompose. Specifically, by construction, for every index $i = 1,2,...,k$, a cluster $C \in S_i$ has at most $O(q \log n) = O(n^{1/k} \log n)$ other clusters $C' \in \cup_{j = i}^k S_j$ at distance at most $\sigma - 1$ from it in $\hat{G}$. This fact is summarized in the next lemma.
\begin{lem} \label{partitionofgraph}
$(S_1,S_2,...,S_k)$ is an $H$-partition with degree $O(n^{1/k} \log n)$ of the supergraph $({\cal G}(Q))^{\sigma - 1}$, where $Q = \cup_{i = 1}^k S_i$.
\end{lem}
By invoking one of the algorithms from Section \ref{sc:refine} for coloring low-arboricity graphs (for which a short low-degree $H$-partition is provided) we can obtain an $O(n^{2/k} \log^2 n \log^{(t)} n)$-labeling for $Q$, which has the property that any two distinct clusters $C,C'$ which receive the same label are at distance at least $\sigma$ from one another in ${\cal G}(Q)$, and thus at distance at least $\sigma$ from one another in $G$. The running time of this step is $O(t \cdot Diam(Q))$. (See Theorem \ref{partdec}.) Alternatively, one can have an $O(n^{2/k} \log n)$-labeling with this property in time $O(Diam(Q) \cdot k^2)$. (See Corollary \ref{improvement}.)

Next we analyze $Diam(Q)$. To do it we first describe Procedure Sep-Partition. (See Algorithm \ref{proced:partitin} for Procedure Partition.) The procedure accepts the same parameters as Procedure Partition, but also the separation parameter $\sigma$. (In fact, Procedure Partition is a special case of Procedure Sep-Partition, where $\sigma = 2$.) Similarly to Procedure Partition, in Procedure Sep-Partition every vertex $v$ selects itself (joins $D$) independently at random with probability $1/q$. Then every selected vertex $v$ sends an exploration message to distance $\sigma - 1$ from it in $\hat{G}$. Every vertex $u$ which is not selected ($u \notin D$) and receives at least one exploration message joins the cluster centered by the closest originator of an exploration message received by $u$. (Ties are broken in an arbitrary but consistent way according to the identities of originators. If originators themselves are clusters, then each of them has its own leader whose identity serves as the identity of the cluster. The consistent rule for breaking ties may be, for example, to prefer an originator with a smaller identity.) Other vertices join the set $A$. The procedure returns the set $A$ and the set ${\cal B}$ of clusters which are created in the way described above. Observe that if $\hat{G}$ is not the original graph but rather a supergraph of it then the algorithm is executed by clusters rather than by single vertices. In other words, in this case the center of each cluster simulates all the operations that need to be performed by the cluster.

The next lemma shows that clusters created by Procedure Sep-Partition are connected and have bounded diameter.
\begin{lem} \label{invokepartition}
Consider an invocation of Procedure Sep-Partition($\hat{G}, q, \sigma$), where $q \geq 1$ is a parameter and $\sigma \geq 2$ is an integer parameter. Then each vertex $v \in A$ has degree $O(q \cdot \log n)$ in $\hat{G}^{(\sigma - 1)}$, and each cluster $C \in {\cal B}$ has (strong) diameter at most $2\sigma - 2$ in $\hat{G}$.
\end{lem}
\begin{proof}
Let $c$ be a sufficiently large fixed constant, and consider a vertex $v \in \hat{G}$ such that a $(\sigma - 1)$-neighborhood $B_{\sigma - 1}(v)$ of $v$ in $\hat{G}$ contains at least $c \cdot q \cdot \log n$ vertices. Then with probability at least $1 - 1/n^c$ at least one of the vertices $u \in B_{\sigma - 1}(v)$ joins $D$, and the vertex $v$ becomes clustered. Hence with probability at least $1- 1/n^{c-1}$ all vertices $v$ with $|B_{\sigma - 1}(v)| \geq c \cdot q \cdot \log n$ become clustered, and so each unclustered vertex $v \in A$ satisfies $|B_{\sigma - 1}(v)| < c \cdot q \cdot \log n$.

Consider a cluster $C \in {\cal B}$. It is centered around an originator $v$ of an exploration message. (The vertex $v$ belongs to $D$, i.e., it is selected.) Consider a vertex $u \in C$, and let $P_{v,u}$ be a shortest $v-u$ path in $\hat{G}$. Let $x$ be a vertex on this path. (Note that $v,u,x$ are vertices of $\hat{G}$, i.e., they are possibly clusters themselves.) It follows that $v$ is the closest selected vertex to $x$, and if there exists another selected vertex $v' \in D$ which satisfies $\mbox{dist}_{\hat{G}}(v,x) = \mbox{dist}_{\hat{G}}(v',x)$, then $v$ has a smaller identity than $v'$. (As otherwise $v'$ would rule $u$ as well.) Hence $x \in C$. Consequently all vertices of $P_{v,u}$ are in $C$, and the length of $P_{u,v}$ is at most $\sigma - 1$. Hence the cluster $C$ has strong radius at most $\sigma - 1$, i.e., strong diameter at most $2(\sigma - 1)$.
\end{proof}
Observe also that by the same argument as in Lemma \ref{partitnscnd}, the number of clusters in ${\cal B}$ is, with high probability, $O(s/q)$. We are now ready to analyze the diameter $Diam(Q)$ of the ultimate network-decomposition $Q$. The following lemma generalizes Lemma \ref{dec}.
\begin{lem}
Let $Q = (S_1,S_2,...,S_k)$ be a $\sigma$-separated network-decomposition produced by the invocation Sep-Decompose($G, n, k, s := n, \epsilon, t, \sigma$) on an input graph $G$. Then for each $i \in [k]$, $Diam(S_i) \leq (2 \sigma - 1)^{ i - 1} - 1$.
\end{lem}
\begin{proof}
We prove by induction on $i$ that in the $i$th level recursive invocation of Procedure Sep-Decompose each vertex $v$ of the input graph $\hat{G}_i$ of this invocation is a cluster of the original graph $G$ with diameter at most $(2\sigma - 1)^{i - 1} - 1$.
Since for each $i \in [k]$, clusters of $S_i$ are vertices of $\hat{G}_i = (\hat{V}_i, \hat{E}_i)$ the assertion of the lemma follows from the inductive claim.\\
{\bf Base:} $Diam(S_1) = 0 = (2\sigma - 1)^0 - 1$.\\
{\bf Step:} Consider an index $i < k$. By Lemma \ref{invokepartition}, each cluster $C$ created by the $i$th level invocation of Procedure Sep-Decompose has strong diameter at most $2\sigma - 2$ in $\hat{G}_i = (\hat{V}_i, \hat{E}_i)$. It follows that
$$ Diam(C) \leq (2 \sigma - 1) \cdot \max_{C' \in \hat{V}_i} \{Diam(C')\} + (2\sigma - 2).$$
By induction hypothesis it follows that
$$Diam(C) \leq (2\sigma - 1)((2\sigma - 1)^{i - 1} -1) +(2\sigma - 2) = (2\sigma - 1)^i - 1.$$
Since vertices $v$ of $\hat{G}_{i + 1}$ are clusters which were formed by the $i$th level invocation of Procedure Sep-Decompose, the assertion of the lemma follows.
\end{proof}

We summarize this discussion with the following corollary.
\begin{col} \label{separatedcol}
Consider an invocation of Sep-Decompose($G, n, k , s:= n, \epsilon, t ,\sigma$), where $k \geq 1, \sigma \geq 2$ are integer parameters. It produces a $\sigma$-separated strong $((2\sigma - 1)^{k - 1} - 1, O(n^{1/k} \log n))$-network-decomposition $Q = \cup_{i = 1}^k S_i$, along with an $H$-partition $(S_1,S_2,...,S_k)$ for $({\cal G}(Q))^{\sigma - 1}$. The running time of this invocation is $O((2\sigma - 1)^{k-1})$.
\end{col}
As was discussed in the paragraph following Lemma \ref{partitionofgraph}, using this network-decomposition one can compute an $O(n^{2/k} \log^2 n \log^{(t)} n)$-labeling for $Q$ within additional $O(t \cdot (2\sigma - 1)^{k-1})$ rounds, or alternatively, an $O(n^{2/k} \log n)$-labeling within additional $O((2 \sigma - 1)^{k-1} \cdot k^2)$ rounds. In both cases the labeling satisfies that any two distinct clusters $C,C'$ which receive the same label are at distance at least $\sigma$ one from another in $G$.


One can also improve the parameters of the network-decomposition from Corollary \ref{separatedcol} from $((2\sigma - 1)^{k - 1} - 1, O(n^{1/k} \log n))$ to $(O(\sigma)^k, O(n^{1/k}))$ at the expense of increasing the running time from $((O(2\sigma - 1)^{k-1})$ to $O(\sigma)^k \log^{2/3} n$ in general graphs, and $O(\sigma)^k \cdot exp \{ O( \sqrt{ \log \log n}) \}$ in graphs with girth at least $6$. This is done by introducing to Procedure Sep-Decompose a modification analogous to the one that we introduced to Procedure Decompose in Section \ref{sc:betterlabels}.
Recall that the difference between Procedure RS-Decompose and Procedure Decompose is that the former invokes Procedure RS-Partition as a subroutine, while the latter invokes Procedure Partition.

Procedure RS-Partition computes a $(3, \delta)$-ruling set $W$ for the set $U = \{ u \in V \ | \ \deg(u) \geq q \}$ of high degree vertices of its input graph $G$, for a parameter $\delta$. The variant of this procedure that we are now describing, called Procedure Sep-RS-Partition, accepts as input also the separation parameter $\sigma$, and computes a $(2\sigma - 1, \delta)$-ruling set $W'$ for the set $U' = \{u \in U \ : \ |B_{\sigma - 1}(u)| \  \geq q \}$ of vertices that have at least $q$ vertices in their $(\sigma - 1)$-ball. The clusters $\{C_w \ | \ w \in W' \}$ are then created in the same way as in Procedure RS-Partition. In particular, their strong radii are still bounded by $\delta$. Also, every vertex $u \in U'$ is assigned to some cluster. The sets $A$ and ${\cal B}$ are now formed as in Procedure RS-Partition. Every vertex $v \in A$ now satisfies $|B_{\sigma - 1}(v)| < q$. The following lemma is analogous to Lemma \ref{rprulin}, and its proof is very similar to that of Lemma \ref{rprulin}.
\begin{lem} \label{rpinv}
Suppose that Procedure Sep-RS-Partition is invoked on a graph $G = (V,E)$ and positive parameters $\sigma$ and $q$. Suppose further that it uses a subroutine for computing a $(2\sigma - 1, \delta)$-ruling set, for a positive integer parameter $\sigma$. Then in the subgraph $G(A)$ every vertex $v \in A$ satisfies $|B_{\sigma - 1}(v)| < q$. Moreover, ${\cal B}$ consists of at most $|V|/q$ supernodes, each of which is a cluster of strong diameter at most $2 \delta$.
\end{lem}
It follows now that Procedure Sep-RS-Decompose computes a $\sigma$-separated $((2\delta + 1)^k, O(n^{1/k}))$-network-decomposition $Q$. For the running time we need again to specify the running time required for computing a $(2\sigma - 1, \delta)$-ruling set. By running the algorithms for computing a ruling set due to Barenboim et al. \cite{BEPS12} and Kothapalli and Pemmaraju \cite{KP12} respectively in $G^{2(\sigma - 1)}$ we obtain a $(2\sigma - 1, 2 \cdot 2\sigma)$-ruling set in the case of general graphs, and a $(2\sigma - 1, 3 \cdot 2\sigma)$-ruling set in the case of graphs of girth at least $6$. In both cases $\delta = O(\sigma)$, and the running time is $O(\sigma \cdot \log^{2/3} n)$ in the former case and $O(\sigma) \cdot exp \{ O(\sqrt{ \log \log n }) \}$ time in the latter.
The rest of the analysis is identical, except that the overall running time of Procedure Sep-RS-Decompose becomes $(O(\sigma))^k \log^{2/3} n$ and $(O(\sigma))^k \cdot exp \{ O(\sqrt {\log \log n}) \}$ in the cases of general graphs and graphs of girth at least $6$, respectively.
\begin{thm}
Procedure Sep-RS-Decompose invoked on an $n$-vertex graph $G = (V,E)$ with positive integer parameters $k$ and $\sigma$ computes a $\sigma$-separated strong $((O(\sigma))^k, n^{1/k})$-network-decomposition $Q$ in randomized time $(O(\sigma))^k \log^{2/3} n$ in general graphs and in $(O(\sigma))^k \cdot exp \{ O(\sqrt{ \log \log n } ) \}$ randomized time in graphs of girth at least $6$.
\end{thm}
One application of strong separated network-decomposition is {\em low-intersecting partitions}. Low-intersecting partitions were introduced by Busch et al. \cite{BDRRS12}, in their work on universal Steiner trees. A {\em low-intersecting $(\alpha, \beta, \gamma)$-partition} ${\cal P}$ of a graph $G$ is the partition of the vertex set $V$ such that \\
(1) Every cluster $C$ in ${\cal P}$ has strong diameter at most $\alpha \cdot \gamma$. \\
(2) For every vertex $v \in V$, a ball $B_{\gamma}(v)$ of radius $\gamma$ around $v$ intersects at most $\beta$ clusters of ${\cal P}$.\\

Busch et al. showed that given a hierarchy of low-intersecting partitions with certain properties (see \cite{BDRRS12} for details) one can construct a universal Steiner tree. (See \cite{BDRRS12} for the definition of universal Steiner tree.) Also, vice versa, given universal Steiner tree they showed that one can construct a low-intersecting partition. They constructed a low-intersecting partition with $\alpha = 4^k, \beta = k \cdot n^{1/k}$, and arbitrary $\gamma$.

We next argue that a $(2\gamma + 1)$-separated strong $(\mu, \eta)$-network-decomposition $Q$ is also a low-intersecting partition with parameters $(\alpha = \mu/\gamma, \beta = \eta, \gamma$). Indeed, every cluster $C$ of $Q$ has strong diameter at most $\mu = \alpha \cdot \gamma$. Moreover, consider a vertex $v$ and a ball $B_{\gamma}(v)$ of radius $\gamma$ around $v$. Observe that for every color class $i \in [\eta]$ of ${\cal G}(Q)$, the ball $B_{\gamma}(v)$ can intersect at most one cluster $C$ colored by $i$. (This is because for every two $i$-colored clusters $C,C'$, it holds that $\mbox{dist}_G(C,C') \geq 2\gamma + 1$.)
Hence altogether $B_{\gamma}(v)$ may intersect up to $\eta$ clusters of $Q$. This proves the claim.

Therefore, our distributed algorithm for computing a $(2\gamma + 1)$-separated strong $(O(\gamma))^k, n^{1/k})$-network-decomposition in distributed randomized time $(O(\gamma))^k \log^{2/3} n$ in general graphs and in \\ $(O(\gamma))^k \cdot exp \{ O(\sqrt {\log \log n} ) \}$ in graphs of girth at least $6$ provides also a distributed algorithm with the same running time for constructing a low-intersecting $((O(\gamma))^k, n^{1/k}, \gamma)$-partition.
We summarize:
\begin{col}
For any pair of positive integer parameters $k, \gamma$, a low-intersecting $((O(\gamma))^k, n^{1/k}, \gamma)$-partition can be constructed in $(O(\gamma))^k \log^{2/3} n$ randomized time in general graphs and in $(O(\gamma))^k \cdot exp \{ O(\sqrt {\log \log n} ) \}$ randomized time in graphs of girth at least $6$.
\end{col}
We remark that this construction can be implemented using short messages. \\ Comparing this result with the algorithm of Busch et al. \cite{BDRRS12} we note that the partition of \cite{BDRRS12} has smaller radius. (It is $\gamma \cdot (O(1))^k$ instead of $(O(\gamma))^k$ in our case.) On the other hand, the intersection parameter $\beta$ of our partitions is smaller. (It is $n^{1/k}$ instead of $k \cdot n^{1/k}$.) In particular, the intersection parameter in the construction of \cite{BDRRS12} is always $\Omega(\log n)$, while ours can be as small as one wishes. Finally, the algorithm of Busch et al. \cite{BDRRS12} is not distributed, and seems inherently sequential.

\subsection{Approximation Algorithms for the Minimum Dominating Set and Minimum $t$-Spanner Problems}
In this section we employ our network-decomposition algorithm in order to derive approximation algorithms for the minimum dominating set and minimum $t$-spanner problems.
Suppose that we are given a $3$-separated $(d,\ell)$-network-decomposition $Q$ of a graph $G$. For each cluster $C \in Q$, we compute in parallel a dominating set $D \subseteq \Gamma^+(C)$ of $C$, such that $D$ has minimum cardinality among all dominating sets $D' \subseteq \Gamma^+(C)$ of $C$. The computation of $D$ is performed by collecting the topology of the clusters and their neighborhoods by the leaders of respective clusters, performing the computation locally using exhaustive search\footnote[1]{We note that once can employ polynomial-time local computations instead of exhaustive search in the expense of increasing the approximation ratio by a factor of $O(\log \Delta)$. See Section \ref{sc:fast}.} , and broadcasting the results to the vertices of the clusters and their neighbors. Since the weak diameter of the clusters is at most $d$, this requires $O(d)$ rounds. We next show that the resulting set obtained by taking the union of the dominating sets in all clusters constitutes an $\ell$-approximate minimum dominating set of the input graph $G$.
\begin{lem} \label{dset}
For a $3$-separated $(d,\ell)$-network-decomposition $Q$, suppose that we have computed a minimum dominating set $D_C \subseteq \Gamma^+(C)$ of $C$, for each cluster $C \in Q$. Then $|\bigcup \{ D_C \ | \ C \in Q \}| \leq \ell \cdot |MDS(G)|$.
\end{lem}
\begin{proof}
For $1 \leq i \leq \ell$, let $U_i \subseteq V$ denote the set of all vertices with label $i$ in the network-decomposition $Q$. Let $\hat{U_i} = \Gamma^+(U_i)$. We claim that $|\bigcup \{D_C \ | \ C \subseteq U_i \}| \leq |W|$,
where $W$ is a minimum dominating set of $G$. (Note that in the current proof the notation $C \subseteq U_i$ stands for a {\em cluster} $C$ that belongs to $U_i$, rather then just a subgraph of $U_i$, since $D_C$ is defined only for clusters.) 
Let $C_i \in Q$ be a cluster of label $i$, $1 \leq i \leq \ell$. 
Then $C_i \subseteq U_i$. Observe that $W \cap \Gamma^+(C_i)$ is a dominating set of $C_i$. (Since $W$ is a dominating set of $C_i$, and any vertex in $W \setminus \Gamma^+(C_i)$ does not dominate any vertex in $C_i$.) Therefore, $|W \cap \Gamma^+(C_i)| \geq |D_{C_i}|$. Note also that for any cluster $C'_i \neq C_i$ of label $i$ it holds that $\Gamma^+(C'_i) \cap \Gamma^+(C_i) = \emptyset$. Indeed, $Q$ is a $3$-separated network-decomposition, and thus, for any $u \in C_i, v \in C'_i$ it holds that $\mbox{dist}_G(u,v) \geq 3$. Hence for any $w \in \Gamma^+(C_i), x \in \Gamma^+(C'_i)$, it holds that $\mbox{dist}_G(x,w) \geq 1$, and thus $x \neq w$. Consequently,  $$|\cup\{D_C \ | \ C \subseteq U_i \}| = \sum_{C \subseteq U_i} |D_C| \leq \sum_{C \subseteq U_i} |W \cap \Gamma^+(C_i)| \leq |W| = |MDS(G)|.$$
Therefore,
$$|\bigcup \{ D_C \ | \ C \in Q \}| = | \bigcup \left \{ (\cup \{ D_C \ | \ C \subseteq U_i \}) : i \in [\ell] \right \} | \leq  \sum_{i = 1}^{\ell} |\cup \{ D_C \ | \ C \subseteq U_i \}|  \leq \ell \cdot |MDS(G)|.$$
\end{proof}

Recall that by Corollary \ref{polylgdecom}, there is a routine that computes an $((O(\log n))^{k - 1}, n^{1/k})$-network-decomposition in deterministic time $(O(\log n))^{k-1}$, for any $k = 1,2,...$. As was discussed above, this routine can also be adapted to compute a weak $3$-separated network-decomposition with the same properties within the same running time. (See Section \ref{sc:strongdecomp}; both the diameter parameter and the running time grow by a constant factor $\sigma = 3$.) Also, similarly to Corollary \ref{separatedcol}, one can adapt this routine so that it will compute a strong network-decomposition with the same parameters and the same running time. (The diameter and the running time grow by a factor of $(2\sigma - 1)^k = 5^k$, which is however swallowed by the notation $(O(\log n))^{k-1}$.) Using this network-decomposition in conjunction with Lemma \ref{dset} we obtain the following theorem.
\begin{thm}
For an $n$-vertex graph $G$, and a positive integer parameter $k$ an $O(n^{1/k})$-approximation for the minimum dominating set problem can be computed in deterministic time $(O(\log n))^{k - 1}$.
\end{thm}
Another problem for which an efficient approximation algorithm can be obtained using network-decompositions is the {\em minimum $t$-spanner} problem. Given an (unweighted) graph $G = (V,E)$ and a positive integer parameter $t$, a subgraph $G' = (V,H)$, $H \subseteq E$, is a {\em $t$-spanner} of $G$ if for every pair $u,v \in V$ of vertices, $\mbox{dist}_{G'}(u,v) \leq t \cdot \mbox{dist}_G(u,v)$. In the {\em minimum $t$-spanner} problem the objective is to find a $t$-spanner of the input graph with as few edges as possible.

Suppose that we are given a $(2t - 1)$-separated $(d, \ell)$-network-decomposition $Q$ of an input graph $G= (V,E)$.
Let $C_1,C_2,...,C_h$ be a single color class of this labeling, i.e., $Diam(C_i) \leq d$ for each $i \in [h]$, and $\mbox{dist}_G(C_i,C_j) \geq 2t - 1$, for every pair of distinct indices $i \neq j$, $i,j \in [h]$. Let $\hat{C}_i = B_{t-1}(C_i)$, for every $i \in [h]$. Note that $\hat{C}_i \cap \hat{C}_j = \emptyset$, for every pair of distinct indices $i \neq j$. Denote ${\cal C} = \cup_{ i =1}^h C_i, \hat{{\cal C}} = \cup_{i=1}^h {\hat C}_i$, and consider a minimum $t$-spanner ${\cal H}$ for $E({\cal C})$ which is allowed to use edges from $E(\hat{{\cal C}})$. Let also $H^*$ be a minimum $t$-spanner for $G$.
\begin{lem} \label{tspan}
$|{\cal H}| \leq |H^*|$.
\end{lem}
\begin{proof}
Observe that the restriction $H^*(\hat{{\cal C}})$ of $H^*$ to $\hat{{\cal C}}$ is a $t$-spanner for ${\cal C}$. Indeed, consider an edge $(u,v) \in E({\cal C})$. Let $C_i \in {\cal C}$ be the cluster such that $u,v \in C_i$. Then $H^*$ contains a path of length at most $t$ between $u$ and $w$, and so this path belongs to $H^*(\hat{C_i}) \subseteq H^*(\hat{{\cal C}})$.

The lemma now follows as ${\cal H}$ is the minimum $t$-spanner for ${\cal C}$ which is allowed to use edges from $E(\hat{C})$, while $H^*(\hat{{\cal C}})$ is a $t$-spanner for ${\cal C}$ of this type. Hence $|{\cal H}| \leq |H^*(\hat{{\cal C}})| \leq |H^*|$.
\end{proof}
Denote also by $H_i$ the minimum $t$-spanner for $E(C_i)$ which is allowed to use edges of $E(\hat{C_i})$.
\begin{lem} \label{tspansetsmall}
$|\cup_{i=1}^h H_i| = |{\cal H}|$.
\end{lem}
\begin{proof}
Obviously, $\cup_{i = 1}^h H_i$ is a $t$-spanner for $E({\cal C})$ which uses only edges of $E(\hat{{\cal C}})$. Hence by optimality of ${\cal H}$, $|\cup_{i = 1}^h H_i| \geq |{\cal H}|$.

In the opposite direction, for every index $i \in [h]$, let ${\cal H}_i = {\cal H} \cap E(\hat{C_i})$. By optimality of $H_i$, $|H_i| \leq |{\cal H}_i|$. Also, for every pair of distinct indices $i,j \in [h]$, ${\cal H}_i \cap {\cal H}_j = \emptyset$. (This is because $E(\hat{C_i}) \cap E(\hat{C_j}) = \emptyset$.)
Hence
$$|{\cal H}| = |\cup_{i = 1}^h {\cal H}_i| = \sum_{i = 1}^h |{\cal H}_i| \geq \sum_{i =1}^h |H_i| \geq |\cup_{i = 1}^h H_i|. $$ 
(The last inequation is, in fact, equality.)
\end{proof}
In other words, to compute a minimum $t$-spanner ${\cal H}$ for $E({\cal C})$ one can compute minimum $t$-spanners $H_1,H_2,...,H_{\ell}$ for $E(C_1),E(C_2),...,E(C_h)$ (which are allowed to use edges of $E(\hat{C_1}), E(\hat{C_2}),..., E(\hat{C_k})$, respectively), and take their union.
Our distributed algorithm will do precisely this. In each cluster $C$ of $Q$ it computes a minimum $t$-spanner for $E(C)$ using edges of $E(\hat{C})$, $\hat{C} = B_t(C)$. This computation is done by collecting the entire topology of $(\hat{C}, E(\hat{C}))$ into a vertex in $C$, doing a local (possibly very heavy) computation, and informing all vertices of $\hat{C}$ about the results of this computation. The union of all these $t$-spanners will be our ultimate spanner. Hence the algorithm returns a spanner ${\cal H}' = \cup_{j = 1}^{\ell} {\cal H}^{(j)}$, where for each index $j \in [h]$, ${\cal H}^{(j)}$ is a minimum $t$-spanner for $E({\cal C}^{(j)})$, where ${\cal C}^{(j)}$ is the set of all vertices labeled by $j$ in the network decomposition $Q$. (In other, words, they belong to clusters of color $j$. Note, however, that to execute the algorithm we do not need to know these colors/labels.) Since by Lemma \ref{tspansetsmall}, for every $j \in [\ell]$, $|{\cal H}^{(j)}| \leq |H^*|$, it follows that the algorithm returns an $\ell$-approximation. The running time of the algorithm is $O(Diam(Q) + t) = exp \{O(k)\} + O(t)$. To summarize:
\begin{thm}
For any pair of positive integer parameters $t,k$ an $O(n^{1/k})$-approximation of the minimum $t$-spanner problem in $n$-vertex graphs can be computed in $exp\{O(k)\} + O(t)$ randomized time.
\end{thm}
Observe that the same result applies to the $t$-spanner problem in {\em directed} graphs, by the same argument. Note that even though the graph is directed, we assume that the communication over every edge is bidirectional.

\section{Removing heavy local computations from the minimum dominating set and minimum $t$-spanner algorithms} \label{sc:fast}
It is well known that an $O(\log \Delta)$-approximation of minimum dominating set can be computed in polynomial time in the sequential setting. (See, e.g., \cite{W04}.) However, this approach cannot be applied directly to our algorithms since we compute minimum dominating sets $D_C$ of clusters $C$, such that $D_C \subseteq \Gamma^+(C)$ rather than $D_C \subseteq C$.
On the other hand, this problem reduces to the Set Cover problem with both the degree parameters (i.e., the maximum cardinality of a set and the maximum number of sets that share an element) bounded by $\Delta + 1$. Hence this problem admits a polynomial-time $O(\log \Delta)$-approximation algorithm. (See, e.g., \cite{SS12}.) One can also extend the classical centralized $O(\log \Delta)$-approximation algorithm for the MDS problem directly to our slightly more general problem. This extension is described below.
Consequently, we can obtain a dominating set whose size is at most $O(\log \Delta)$ the size of the minimum dominating set of $C$ consisting of vertices of $\Gamma^+(C)$. This can be achieved in the following way. Initially $D_C = \emptyset$. We proceed in phases, each time selecting a vertex $v$ from $C$ such that $d_v = |\Gamma^+(v) \cap C \setminus \Gamma^+(D_C)|$ is maximal, and adding $v$ to $D_C$.  (Ties are broken by preferring vertices that belong to $\Gamma^+(D_C)$, and if this does not solve the tie, it is broken arbitrary.)  Once no uncovered vertex remains we are done.

Let $S^* \subseteq \Gamma^+(C)$ be a minimum dominating set of $C$. We claim that $D_C \leq O(\log \Delta \cdot |S^*|)$. The proof is by amortized analysis. Each time a vertex $v$ is added to $D_C$ we assign a weight $1/d_v$ to each vertex of $\Gamma^+(v) \cap C \setminus \Gamma^+(D_C)$. Observe that the sum of all weights assigned during this procedure is $|D_C|$. Next, (for analysis) let each vertex of $C$ select a single vertex from $S^*$ that dominates it. Consider a vertex $u \in S^*$ and the set $W$ of all neighbors of $u$ in $C$ that selected $u$.
Next, we analyze the sum of weights of $W$. For each $w \in W$ it is assigned a weight once a neighbor of $w$ (or $w$ itself) joins the dominating set $D_C$. Let $i$ be the number of the phase in which it happens, and $\deg_i(u)$ denote the number of neighbors of $u$ in $C$ that are not covered in the beginning of phase $i$. Also, let $z$ denote the neighbor that dominates $w$, for which $w$ obtained its weight. Since in each phase a vertex $v$ with maximal $d_v$ is selected, it holds that $\deg_i(u) \leq \deg_i(z)$. Consequently, $w$ is assigned a weight at most $1/\deg_i(z) \leq 1/\deg_i(u)$. Therefore, the sum of weights of $W$ is at most $\sum_{j = 1}^{|W|} 1/j = O(\log \Delta)$. Therefore, the sum of all weights in the graph is $|D_C| = O(\log \Delta \cdot |S^*|)$. This completes the proof. 

A similar idea can be applied in the case of the minimum $t$-spanner problem. Again, we need a centralized polynomial-time approximation algorithm for the minimum $t$-spanner for edges of $E(C)$ (for a cluster $C$), while the spanner is allowed to use edges of $E(\hat{C})$. This is an instance for the {\em client-server $t$-spanner} problem, and for the case $t = 2$ it was devised in \cite{EP01}. By plugging it in our distributed algorithm for approximating spanners we obtain a distributed $O(n^{1/k} \log n)$-approximation algorithm with running time $exp\{O(k)\}$ for the directed and undirected $2$-spanner problem. The latter algorithm only employs polynomially-bounded local computations. To the best of our knowledge, there are no existing centralized algorithms with a non-trivial approximation guarantee for the client-server $t$-spanner problem for $t \geq 3$. It is however likely that the LP-based approaches to the minimum $t$-spanner problem (such as \cite{BBMRY11,DK11}) extend to this more general problem.

\end{document}